\documentclass[pagebackref,acmsmall,10pt,screen,nonacm]{acmart}
\usepackage[utf8]{inputenc}
\usepackage[subpreambles=false]{standalone}
\usepackage{xcolor}
\usepackage[most]{tcolorbox}
\usepackage{mathpartir}
\usepackage{bcprules}
\usepackage{tcolorbox}
\usepackage{booktabs}
\usepackage{import}
\usepackage{multicol}
\usepackage{listings}
\usepackage{listings-rust}
\usepackage{listings-scheme}
\usepackage{xspace}

\providecommand{\sub}{}
\providecommand{\comment}{}

\providecommand{\Fm}{}
\providecommand{\Fq}{}
\providecommand{\Fa}{}
\providecommand{\Fc}{}

\providecommand{\judgement}[2]{}
\providecommand{\box}{}
\providecommand{\ruledef}[1]{}
\providecommand{\ruledefN}[2]{}
\providecommand{\ruledefR}[2]{}
\providecommand{\ruleref}[1]{}
\providecommand{\rulerefN}[2]{}
\providecommand{\gap}{}

\providecommand{\comment}{}
\providecommand{\unboxed}{}
\providecommand{\setboxed}[2]{}

\providecommand{\highlight}[1]{}
\providecommand{\reduces}{}
\providecommand{\Fsub}{}
\providecommand{\ruleref}[1]{}

\providecommand{\fn}[3]{}
\providecommand{\fna}[3]{}
\providecommand{\Sfna}[4]{}
\providecommand{\Qfna}[4]{}
\providecommand{\Sfn}[3]{}
\providecommand{\Qfn}[3]{}
\providecommand{\app}[3]{}
\providecommand{\capp}[3]{}
\providecommand{\apphighlight}[3]{}
\providecommand{\Sapp}[2]{}
\providecommand{\Qapp}[2]{}
\providecommand{\mapstoty}{}
\providecommand{\mapstote}{}

\providecommand{\qual}[1]{}

\providecommand{\fntype}[2]{}
\providecommand{\cfntype}[2]{}
\providecommand{\Sfntype}[3]{}
\providecommand{\Qfntype}[3]{}
\providecommand{\qtyp}[2]{}
\providecommand{\cks}[2]{}

\renewcommand{\comment}[1]{}

\renewcommand{\sub}{~\texttt{<:}~}
\renewcommand{\Fq}{System $\texttt{F}_{\texttt{<:Q}}$\xspace}
\renewcommand{\Fa}{System $\texttt{F}_{\texttt{<:QA}}$\xspace}
\renewcommand{\Fm}{System $\texttt{F}_{\texttt{<:QM}}$\xspace}
\renewcommand{\Fc}{System $\texttt{F}_{\texttt{<:QC}}$\xspace}

\renewcommand{\Fsub}{System $\texttt{F}_{\texttt{<:}}$\xspace}

\renewcommand{\judgement}[2]{{\bf\textsf{#1}} \hfill #2}
\newcounter{RuleRef}
\renewcommand{\ruledef}[1]{%
  \refstepcounter{RuleRef}\textsc{#1}\label{rule:#1}}
\renewcommand{\ruledefN}[2]{%
  \refstepcounter{RuleRef}\textsc{#2}\label{rule:#1}}
\renewcommand{\ruledefR}[2]{%
  \refstepcounter{RuleRef}#2\label{rule:#1}}
\renewcommand{\ruleref}[1]{\textsc{(\hyperref[rule:#1]{#1})}}
\renewcommand{\rulerefN}[2]{\textsc{(\hyperref[rule:#1]{#2})}}
\renewcommand{\gap}{\quad\quad}

\DeclareMathOperator{\Boxed}{\texttt{box}}
\DeclareMathOperator{\async}{\texttt{async}}
\DeclareMathOperator{\sync}{\texttt{sync}}
\DeclareMathOperator{\fv}{\texttt{fv}}
\DeclareMathOperator{\readonly}{\texttt{readonly}}
\DeclareMathOperator{\mutable}{\texttt{mutable}}
\DeclareMathOperator{\kbarrier}{\texttt{barrier}}
\DeclareMathOperator{\karg}{\texttt{arg}}
\DeclareMathOperator{\kapp}{\texttt{app}}
\DeclareMathOperator{\ktarg}{\texttt{targ}}
\DeclareMathOperator{\kqarg}{\texttt{qarg}}

\DeclareMathOperator{\tracked}{\texttt{tracked}}
\DeclareMathOperator{\eval}{\texttt{eval}}
\definecolor{light-gray}{gray}{0.95} 
\renewcommand{\highlight}[1]{\colorbox{light-gray}{$\displaystyle #1$}}
\renewcommand{\unboxed}{\texttt{unbox}~}
\renewcommand{\setboxed}[2]{\texttt{set-box!}~{#1}~{#2}}

\renewcommand{\reduces}{\;\longrightarrow\;}
\renewcommand{\ruleref}[1]{\textsc{(\hyperref[rule:#1]{#1})}}
\renewcommand{\fn}[2]{\lambda ({#1}) . {#2}}
\renewcommand{\fna}[3]{\lambda ({#2})_{#1} . {#3}}
\renewcommand{\Sfn}[3]{\Lambda ({#1} \sub {#2}) . {#3}}
\renewcommand{\Sfna}[4]{\Lambda ({#2} \sub {#3})_{#1} . {#4}}
\renewcommand{\Qfn}[3]{\Lambda ({#1} \sub {#2}) . {#3}}
\renewcommand{\Qfna}[4]{\Lambda ({#2} \sub {#3})_{#1} . {#4}}
\renewcommand{\app}[2]{{#1}({#2})}
\renewcommand{\apphighlight}[3]{{#1}\highlight{[{#2}]}({#3})}
\renewcommand{\capp}[3]{{#1}\{\!\!\{{#2}\}\!\!\}({#3})}
\renewcommand{\Sapp}[2]{{#1}[{#2}]}
\renewcommand{\Qapp}[2]{{#1}\{\!\!\{{#2}\}\!\!\}}
\renewcommand{\fntype}[2]{{#1} \to {#2}}
\renewcommand{\cfntype}[3]{({#1} : {#2}) \to {#3}}
\renewcommand{\Sfntype}[3]{\forall({#1} \sub {#2}).{#3}}
\renewcommand{\Qfntype}[3]{\forall({#1} \sub {#2}).{#3}}
\renewcommand{\qtyp}[2]{\{{#1}\}~{#2}}
\renewcommand{\cks}[2]{\langle {#1}, {#2}\rangle}
\definecolor{qualifier-blue-bg}{RGB}{233, 239, 243}
\definecolor{qualifier-blue}{RGB}{45, 95, 134}
\renewcommand{\qual}[1]{\colorbox{qualifier-blue-bg}{\color{qualifier-blue}{\tt #1}}}
\DeclareMathOperator{\upqual}{\texttt{upqual}}
\DeclareMathOperator{\assert}{\texttt{assert}}

\renewcommand{\mapstoty}{\mapsto_{\tt type}}
\renewcommand{\mapstote}{\mapsto_{\tt term}}

\usepackage{array}

\acmJournal{PACMPL}
\acmVolume{1}
\acmNumber{OOPSLA}
\acmArticle{1}
\acmYear{2018}
\acmMonth{1}
\acmDOI{} 
\setcopyright{none}
\startPage{1}

\lstdefinelanguage{scala}{
  morekeywords={%
          abstract,case,catch,class,def,do,else,extends,%
          false,final,finally,for,forSome,if,implicit,import,lazy,%
          match,new,null,object,override,package,private,protected,%
          return,sealed,super,this,throw,trait,true,try,type,%
          val,var,while,with,yield},
  otherkeywords={=>,<-,<\%,<:,>:,\#,@},
  sensitive=true,
  morecomment=[l]{//},
  morecomment=[n]{/*}{*/},
  morestring=[b]",
  morestring=[b]',
  morestring=[b]"""
}[keywords,comments,strings]

\begin{document}
\title{Qualifying \Fsub}
\author{Edward Lee}
\orcid{0000-0001-7057-0912}
\affiliation{
   \department{Computer Science}              
   \institution{University of Waterloo}            
   \streetaddress{200 University Ave W.}
   \city{Waterloo}                                                     
   \state{ON}
   \postcode{N2L 3G1}
   \country{Canada}                    
 }
\author{Yaoyu Zhao}
\affiliation{
   \department{Computer Science}              
   \institution{University of Waterloo}            
   \streetaddress{200 University Ave W.}
   \city{Waterloo}                                                     
   \state{ON}
   \postcode{N2L 3G1}
   \country{Canada}                    
 }
\author{Ondřej Lhoták}
\orcid{0000-0001-9066-1889}
\affiliation{
   \department{Computer Science}              
   \institution{University of Waterloo}            
   \streetaddress{200 University Ave W.}
   \city{Waterloo}                                                     
   \state{ON}
   \postcode{N2L 3G1}
   \country{Canada}                    
 }
\author{James You}
\orcid{0009-0000-5906-0305}
\affiliation{
   \department{Computer Science}              
   \institution{University of Waterloo}            
   \streetaddress{200 University Ave W.}
   \city{Waterloo}                                                     
   \state{ON}
   \postcode{N2L 3G1}
   \country{Canada}                    
 }
\author{Kavin Satheeskumar}
\orcid{0009-0002-1106-2429}
\affiliation{
   \department{Computer Science}              
   \institution{University of Waterloo}            
   \streetaddress{200 University Ave W.}
   \city{Waterloo}                                                     
   \state{ON}
   \postcode{N2L 3G1}
   \country{Canada}                    
 }

\author{Jonathan Brachthäuser}
\orcid{0000-0001-9128-0391}
\affiliation{
   \department{Computer Science}              
   \institution{University of Tübingen}            
   \streetaddress{Sand 13}
   \city{Tübingen}                                                     
   \state{BaWü}
   \postcode{72076}
   \country{Germany}                    
}
\date{January 2023}

\begin{abstract}
Type qualifiers offer a lightweight mechanism for enriching existing type systems to enforce additional, desirable, program invariants.  
They do so by offering a restricted but effective form of subtyping.  
While the theory of type qualifiers is well understood and present in many programming languages today, polymorphism over type qualifiers is an area that is less examined.  
We explore how such a polymorphic system could arise by constructing a calculus \Fq which combines the higher-rank bounded polymorphism of \Fsub with the theory of type qualifiers. We explore how the ideas used to construct \Fq can be reused in situations where type qualifiers naturally arise---in reference immutability, function colouring, and capture checking.  Finally, we re-examine other qualifier systems in the literature in light of the observations presented while developing \Fq.
\end{abstract}



\maketitle
\tcbset{
  on line,
  boxsep=4pt,
  left=0pt,
  right=0pt,
  top=0pt,
  bottom=0pt,
  boxrule=0pt, 
  colframe=white,
  colback=qualifier-blue-bg,
  coltext=qualifier-blue, 
  highlight math style={enhanced},
}

\section{Introduction}
Static type systems classify the values a program reduces to. For example, the signature of the function

\begin{lstlisting}[language=Scala]
    def toLowerCase(in: String): String = { ... }
\end{lstlisting}

\noindent
{\tt toLowerCase} enforces that it takes in a {\tt String} as an argument and returns a {\tt String} as a result.  If strings are implemented as mutable heap objects, how would we express the additional property that {\tt toLowerCase} does not its mutate its input? 

There are at least two ways to address this.  We can view the modification of {\tt toLowerCase}'s argument {\tt in} as a property of {\tt toLowerCase} or we can view mutability as a property of the argument string {\tt in} itself.  The former viewpoint leads to solutions like (co-)effect systems \cite{10.1145/2628136.2628160} that describe the relation of a function to the context it is called in.
The latter viewpoint, of viewing it as a property of the argument, leads to systems that enrich the types of values with additional information. In this paper, we adopt the latter view.

{\it Type qualifiers} by \citet{10.1145/301618.301665} is one such system. In such a system, we could qualify the type of {\tt toLowerCase}'s argument with the type qualifier \qual{const} to express that {\tt toLowerCase} cannot modify its argument.  We may choose to annotate its result with the type qualifier \qual{const} to indicate that
its result is a \qual{const} {\tt String} which cannot be changed by {\tt toLowerCase}'s caller.  
\begin{lstlisting}[language=Scala]
    def toLowerCase(in: const String): const String = {...}
\end{lstlisting}

\noindent
The function {\tt toLowerCase} now accepts an immutable {\tt String} as an argument and presumably returns a new {\tt String} that is a copy of its argument except in lowercase. 
More importantly, since the input string is qualified as {\tt const}, we know that this version {\tt toLowerCase} cannot mutate the input string; for example, such as calling a method like {\tt in.setCharAt(0, 'A')}, which would replace the character of index {\tt 0} of the string with the character {\tt A}.

Perhaps this is too restrictive.  After all, {\tt toLowerCase} will allocate a new {\tt String} and does not impose invariants on it; its caller should be permitted to mutate the value returned.  We should instead annotate {\tt toLowerCase} as follows, with a \qual{mutable} qualifier on its return value. 
\begin{lstlisting}[language=Scala]
    def toLowerCase(in: const String): mutable String = {...}
\end{lstlisting}

\noindent Subtyping naturally arises in this context---a {\tt \qual{mutable} String} can be a subtype of {\tt \qual{const} String}; this change will not alter the semantics of existing calls to {\tt toLowerCase} to break.

Similarly, it would be impractical if {\tt toLowerCase} only accepted immutable {\tt String}s.  After all, any operation one
could perform on a {\it immutable} {\tt String} one should be semantically valid on a {\it mutable} {\tt String} as well.  
Therefore a {\tt \qual{mutable} String} should ideally be a subtype of {\tt \qual{const} String}. If we wanted to, we should be to chain calls to {\tt toLowerCase}!
\begin{lstlisting}[language=Scala]
    toLowerCase(toLowerCase("HELLO WORLD")) == "hello world"
\end{lstlisting}

\citet{10.1145/301618.301665} were the first to recognize this natural subtyping relation induced by type qualifiers, which permitted type qualifiers to be integrated easily into existing type systems with subtyping.  
Perhaps the most well known qualifier is \qual{const}.
\qual{const} is used to mark particular values as read-only or {\it immutable} and it is found in many languages and language extensions \cite{DBLP:books/daglib/0019344,10.1145/3386323,10.1145/1103845.1094828}. 
Other languages, such as OCaml and Rust, are exploring more exotic qualifiers to encode properties like locality, linearity, exclusivity, and synchronicity \cite{Slater_2023_June,Slater_2023,Wuyts_Scherer_Matsakis}.
Qualifiers are so easy to use that many type system extensions start as type qualifier annotations on existing types;
for Java there is a framework \cite{10.1145/1390630.1390656} for doing so, and it has been used to model extensions to Java for checking {\it nullability}, {\it energy consumption}, and {\it determinism} amongst others using type qualifiers.

While type qualifiers themselves are well-explored, {\it qualifier polymorphism} is still understudied.  
Sometimes parametric polymorphism is not necessary when subtyping is present.  
For example, the type signature that we gave to {\tt toLowerCase}, {\tt \qual{const} String => \qual{mutable} String} is indeed the most permissive type that may be assigned.  
In languages with subtyping, variables are only necessary to relate types and qualifiers in both co- and contravariant positions; otherwise we can use their respective type bounds~\cite[Chapter~4.2.1]{Dolan}.  
For example, while we could have made {\tt toLowerCase} polymorphic using a qualifier variable {\tt Q} over the immutability of its input, such a change is unnecessary as we can simply replace {\tt Q} with its upper bound \qual{const} to arrive at the monomorphic but equally general version of {\tt toLowerCase} from above.
\begin{lstlisting}[language=Scala]
    def toLowerCase[Q <: const](in: Q String): mutable String = {...}
\end{lstlisting}

However, variables are {\it indeed} necessary when relating types and qualifiers in covariant positions to types and qualifiers in contravariant positions.  For example, consider a {\tt substring} function.  Which qualifiers should we assign its arguments and return value?
\begin{lstlisting}[language=Scala]
    def substring(in: ??? String, from: Int, to: Int): ??? String = {...}
\end{lstlisting}

\noindent
Clearly a {\tt substring} of an immutable string should itself be immutable, but also a {\tt substring} of a mutable string should be mutable as well.  To express this set of new constraints, we need {\it parametric qualifier polymorphism}.
\begin{lstlisting}[language=Scala]
    def substring[Q <: const](in: Q String, from: Int, to: Int): Q String
\end{lstlisting}

We also need to consider how qualifier polymorphism interacts with type polymorphism.  For example,
what should be the type of a function like {\tt slice}, which returns a subarray of an array?
It needs to be parametric over the the type of the elements stored in the array, where the element type itself could be qualified.  
This raises the question---should type variables range over unqualified types or both unqualified {\it and} qualified types? 
Foster's original system does not address this issue, and existing qualifier systems disagree on what type variables range over and whether or not type variables can be qualified at all.
For reasons we will demonstrate later in Section \ref{section:poly}, type variables should range over unqualified types;
to achieve polymorphism over both types and qualifiers, we need both type variables and qualifier variables for orthogonality.

\begin{lstlisting}[language=Scala]
    def slice[Qa<:const, Qv<:const, T<:Any](in: Qa Array[Qv T]): Qa Array[Qv T]
\end{lstlisting}

Another underexplored area is that of {\it merging} type qualifiers, especially in light of parametric qualifier polymorphism. 
For example, consider the type qualifiers \qual{throws} and \qual{noexcept}, expressing that a function may throw an exception or that it does not throw any exception at all. 
Without polymorphism, it is easy to combine qualifiers. 
For example, a function like {\tt combined}, that calls both pure and exception-throwing functions should be qualified
with the union of the two qualifiers, \qual{throws}, expressing that an exception could be thrown from the calling function.
\begin{lstlisting}[language=Scala]
    def pure() = 0                              // (() => Unit) noexcept
    def impure() = throw new Exception("Hello") // (() => Unit) throws
    def combined() = { pure(); impure() }       // (() => Unit) throws 
\end{lstlisting}

\noindent Things are more complicated in the presence of qualifier parametric higher-order functions, such as:
\begin{lstlisting}[language=Scala]
    def compose[A,B,C,Qf,Qg](f: (A => B) Qf, g: (B => C) Qg)): (A => C) ???
      = (x) => g(f(x))  
\end{lstlisting}
What should be the qualifier on the return type {\tt (A => C)} of the function?
Intuitively, if either {\tt f} or {\tt g} throws an exception, then the result of {\tt compose} should be qualified with \qual{throws}, but if neither throws any exception, then the composition should be qualified with \qual{noexcept}. 
Ideally we would like some mechanism for specifying the {\it union} of the qualifiers annotated on both {\tt f} and {\tt g}.
\begin{lstlisting}[language=Scala]
    def compose[A,B,C,Qf,Qg](f: (A => B) Qf, g: (B => C) Qg)): (A => C) {Qf | Qg}
\end{lstlisting}

\noindent Existing qualifier systems today have limited support for these use cases.
\citet{10.1145/301618.301665}'s original system is limited to simple ML-style qualifier polymorphism with no mechanism for specifying qualifier-polymorphic function types, and has limited support for combining qualifiers.  
Systems that do support explicit qualifier polymorphism like that of \citet{10.1145/2384616.2384619} partially ignore the interaction between combinations of qualifier variables and their bounds, or present application-specific {\it subqualification} semantics seen in \citet{10.1145/3618003} or \citet{wei2023polymorphic}.  Must this always be the case?  Is there something in common we can
generalize and apply to give a {\it design recipe} for designing qualifier systems with subqualification and polymorphism?

We believe this does not need to be the case; we show that it {\it is} possible to add qualifier polymorphism without {\it losing the natural lattice structure of type qualifiers}, and 
that there is a natural way to reconcile type polymorphism with qualifier polymorphism as well.

To illustrate these ideas, we start by first giving a {\it design recipe} for constructing a 
qualifier-polymorphic enrichment \Fq of \Fsub, much in the same way \citet{10.1145/301618.301665} gives a design recipe for adding qualifiers to a base simply-typed lambda calculus.  Our recipe constructs a calculus with the following desirable properties:
\begin{itemize}
    \item {\bfseries\sffamily Higher-rank qualifier and type polymorphism:}  We show how
        to add higher-rank qualifier polymorphism to a system with higher-rank type polymorphism in Section \ref{section:fq}.
    \item {\bfseries\sffamily Natural subtyping with qualifier variables:} We show that 
        the subtyping that type qualifiers induce extends naturally even when working with
        qualifier variables.  We achieve this by using the {\it free lattice} generated
        over the original qualifier lattice.  We illustrate these ideas, first in a simplified context over a fixed two-point qualifier lattice in Section \ref{section:fq}
        and generalize to an arbitrary bounded qualifier lattice in Section \ref{section:multiple}.
    \item {\bfseries\sffamily Easy meets and joins:} As we generalize the notion of a qualifier
        to that of an element from the free (qualifier) lattice, we recover the ability to combine qualifiers using meets and joins.
\end{itemize}

Next, to demonstrate the applicability of our qualifier polymorphism design recipe, we show how one can model three natural problems -- {\it reference immutability}, {\it function colouring}, and {\it capture tracking}, using the ideas used to develop \Fq in Section \ref{section:applications}.  
We then discuss how type polymorphism can interact with qualifier polymorphism in Section \ref{section:poly} to justify our design choices.
We then re-examine a selection of other qualifier systems in light of our observations developed in our free lattice-based subqualification recipe in Section \ref{section:revisit} to see how their subqualification rules fit in our free lattice based design recipe.  Finally, we close with a discussion of other related work in Section \ref{section:related}.

Our soundness proofs are mechanized in the Coq proof assistant; details are discussed in Section~\ref{section:mechanization}.

\section{Qualified Type Systems}
\label{section:qualifier}
In this section, we introduce \Fq, a simple calculus with support for qualified types as well as type- and qualifier polymorphism.
We start off with a brief explanation of what type qualifiers are (Subsection \ref{section:simple}), introduce \Fq (Subsection \ref{section:fq}), and show that it satisfies the standard soundness theorems (Subsection \ref{section:metatheory}).

\subsection{A Simply-Qualified Type System}
\label{section:simple}
    As \citet{10.1145/301618.301665} observes, type qualifiers induce a simple, yet highly useful form of subtyping on qualified types.
    Consider a qualifier like \qual{const}, which qualifies an existing type to be read-only.  It comes
    equipped with a dual qualifier \qual{mutable} which qualifies an existing type to be mutable.     
    The type {\tt \qual{const} T} is a \emph{supertype} of {\tt \qual{mutable} T}, for all types {\tt T}; a mutable value can be used wherever an immutable value is expected. Other qualifier pairs induce a \emph{subtype}, like \qual{noexcept} and \qual{throws}---it is sound to use a function which throws no exception in a context which would handle exceptions.
    Figure \ref{figure:qual:examples} provides an overview of some qualifiers and describes which invariants they model.

    \begin{figure}
        \begin{center}
        \begin{tabular}{m{4cm} m{9cm}}
         {\bfseries\sffamily Qualifiers} & {\bfseries\sffamily Description} \\ \toprule
         {\tt \qual{mutable} \sub \qual{const}} & Mutability; a mutable value could be used anywhere an immutable value is expected.  A {\it covariant} qualifier, as \qual{mutable} is often omitted. \\ \midrule
         {\tt \qual{noexcept} \sub \qual{throws}} & Exception safety; a function which throws no exceptions can be called
            anywhere a function which throws could. A {\it contravariant} qualifier, as \qual{throws}
            is often omitted. \cite{maurer-2015} \\ \midrule
         {\tt \qual{sync} \sub \qual{async}} & Synchronicity; a function which is synchronous and does not suspend can be
            called in contexts where a function which is asynchronus and suspends could. {\it Covariant}, as \qual{sync} is assumed by default.  \\ \midrule
         {\tt \qual{nonnull} \sub \qual{nullable}} & Nullability; a value which is guaranteed not to be null can
            be used in a context which can deal with nullable values.  {\it Covariant}, in systems
            with this qualifier -- most values ought not to be null.\\
        \end{tabular}
        \end{center}        
        \caption{Examples of type qualifiers}
        \label{figure:qual:examples}
    \end{figure}

    Often one of the two qualifiers is assumed by omission -- for example \qual{mutable} and \qual{throws} are often omitted; references are assumed to be mutable unless otherwise specified,
    and similarly functions are assumed to possibly throw exceptions as well.  Qualifiers like \qual{const} where the smaller qualifier is omitted are {\it positive}, or {\it covariant}; by example, {\tt \qual{const} String} is a subtype of a unqualified {\tt String}.
    Conversely, qualifiers like \qual{noexcept} are {\it negative}, or {\it contravariant};
    {\tt String => String \qual{noexcept}} is a subtype of {\tt String => String}.

\subsection{Qualifying a Language}
    The observation that qualifiers induce subtyping relationships allows language designers to seamlessly integrate support for type qualifiers into existing languages with subtyping.
    As \citet{10.1145/301618.301665} point out, these qualifiers embed into a qualifier lattice structure $\mathcal{L}$, and they give a design recipe for enriching an existing type system
    with support for type qualifiers.
    \begin{enumerate}
        \item First, embed qualifiers into a lattice $\mathcal{L}$.  For example,
              \qual{const} and \qual{mutable} embed into a two-point lattice,
              where \qual{const} is $\top$ and \qual{mutable} is $\bot$.  Other example qualifiers
              (and their embeddings) are described in Figure \ref{figure:qual:examples}.
        \item Second, extend the type system so that it operates on {\it qualified types} --
            a pair $\qtyp{l}{T}$ where $l$ is a qualifier lattice element and $T$ a base type
            from the original system.  This is done in two steps.
        \item Embed qualifiers into the subtyping system.  Typically, for two
              qualified types $\qtyp{l_1}{T_1}$ and $\qtyp{l_2}{T_2}$ such that $l_1 \sqsubseteq l_2$ and
              $T_1 \sub T_2$ one will add the subtyping rule $\qtyp{l_1}{T_1} \sub \qtyp{l_2}{T_2}$.
        \item Add rules for {\it introducing qualifiers}, typically in the introduction forms
              for typing values.
        \item Finally, augment the other typing rules, typically elimination forms,
                so that qualifiers are properly accounted for.  One may also additionally
                add an {\it assertion rule} for statically checking qualifiers as well.
   \end{enumerate}

\subsection{Higher-rank Polymorphism}
\label{section:fq}

Foster's original work allows one to add qualifiers to an existing type system. As we discussed earlier, we want more, though:\begin{enumerate}
    \item {\bfseries\sffamily Qualifier Polymorphism:}  Certain functions ought to be polymorphic in the
        qualifiers they expect.  For example, from our introduction,
        we should be able to express a {\tt substring} function which is polymorphic in the mutability of the string passed to it.  While this is easy enough, as \citet{10.1145/301618.301665} shows, the interaction
        of lattice operations with qualifier variables is not so easy, as we discuss below.

    \item {\bfseries\sffamily Merging Qualifiers:}  We often need to merge qualifiers when constructing more
        complicated values.  Merging is easy when working with a lattice; we can just take the lattice's underlying join ($\sqcup$) or meet ($\sqcap$) operation.  But how do we reason about meets or joins of qualifier variables?  For example, in a \qual{noexcept} qualifier system we should be able to collapse the qualifier on the result of a function like {\tt twice} which composes a function with itself from $\qual{Q} \sqcup \qual{Q}$ to just \qual{Q}; the result of {\tt twice} throws if {\tt f} throws or if {\tt f} throws,
        which is namely just if {\tt f} throws.
\begin{lstlisting}[language=Scala]
    def twice[A, Q](f: (A => A) Q): (A => A) Q = compose(f, f)
\end{lstlisting} 
\end{enumerate}

To achieve this, we need to extend qualifiers from just elements of a two-point lattice, as in \citet{10.1145/301618.301665},
to formulas over lattices which can involve qualifier variables in addition to elements of the original lattice.
Moreover, we would like to relate these formulas as well.

As \citet{43df3167-5a81-387d-88d7-2d29cdf1c881} observed, there is a lattice which encodes these relations over these lattice formulas, namely, the {\it free lattice} constructed over the original qualifier lattice.  Free lattices capture exactly the lattice formulas inequalities that are true in every lattice; given two lattice formulas over a set of variables $f_1[\overline{X}] \sqsubseteq f_2[\overline{X}]$ in the free lattice, 
$f_1[\overline{X} \to \overline{L}] \sqsubseteq f_2[\overline{X} \to \overline{L}]$ in every lattice $\mathcal{L}$ and instantiation $\overline{L}$ of the variables in $\overline{X}$ to elements of $\mathcal{L}$.

It should not be surprising to see free lattices here; as \citet[Chapter 3]{Dolan} observed, free lattices can be used to model subtyping lattices with unions, intersections, and variables as well.  This allows us to generalize \citet{10.1145/301618.301665}'s recipe for qualifying types.  Instead of qualifying types by elements of the qualifier lattice, we qualify types by elements of the {\it free lattice} generated over that base qualifier lattice, and we support qualifier polymorphism explicitly with bounds following \Fsub instead
of implicitly at prenex position with constraints as \citet{10.1145/301618.301665} do.

\subsection{\Fq}

\begin{figure}
  \begin{minipage}{0.47\textwidth}
  \[
  \begin{array}[t]{rll@{\hspace{4mm}}l}\\
   s, t & ::= & & \mbox{\bf\textsf{Terms}} \\
        & | & \highlight{\fna{P}{x}{t}} & \mbox{term abstraction} \\
        & | & x & \mbox{term variable} \\
        & | & \app{s}{t} & \mbox{application} \\
        
        & | & \highlight{\Sfna{P}{X}{S}{t}} & \mbox{type abstraction} \\
        & | & \highlight{\Qfna{P}{Y}{Q}{t}} & \mbox{qualifier abstraction} \\
        & | & \Sapp{s}{S} & \mbox{type application} \\
        & | & \highlight{\Qapp{s}{Q}} & \mbox{qualifier application} \\[\medskipamount]

  \Gamma & ::= & & \mbox{\bf\textsf{Environment}} \\
         & | & \cdot & \mbox{empty} \\
         & | & \Gamma,~x : T & \mbox{term binding} \\
         & | & \Gamma,~X <: S & \mbox{type binding} \\
         & | & \Gamma,~\highlight{Y <: Q} & \mbox{qualifier binding} \\

  \end{array}
  \]
  \end{minipage}\hfill
  \begin{minipage}{0.47\textwidth}
  \[
  \begin{array}[t]{rll@{\hspace{4mm}}l}\\
   S    & ::= & & \mbox{{\bf\textsf{Simple Types}}} \\
          & | & \top & \mbox{top type} \\
          & | & \fntype{T_1}{T_2} & \mbox{function type} \\
          & | & X & \mbox{type variable} \\
          & | & \Sfntype{X}{S}{T} & \mbox{for-all type} \\
          & | & \highlight{\Qfntype{Y}{Q}{T}} & \mbox{qualifier for-all type}\\[\medskipamount]
          
   T    & ::= & & \mbox{{\bf\textsf{Qualified Types}}} \\
        & |   & \highlight{\qtyp{Q}{S}} & \mbox{qualified type} \\[\medskipamount]

   P, Q, R & ::= & & \mbox{{\bf\textsf{Qualifiers}}} \\
        & |   & \top, \bot & \mbox{Top and bottom} \\
        & |   & Y & \mbox{Qualifier variables} \\
        & |   & Q \wedge R ~|~ Q \vee R & \mbox{Meets and joins} \\[\medskipamount]

  \end{array}
  \]
  \end{minipage}

  \begin{center}
  \begin{minipage}{0.47\textwidth}
  \[
    \begin{array}[t]{rll@{\hspace{4mm}}l}                
    v   & ::= & & \mbox{\bf\textsf{Runtime Values}} \\
        & | & \fna{P}{x}{t} & \\
        & | & \Sfna{P}{X}{S}{t} \\
        & | & \Qfna{P}{Y}{Q}{t}
    \end{array}
  \]
  \end{minipage}
  \begin{minipage}{0.47\textwidth}
  \[
    \begin{array}[t]{rll@{\hspace{4mm}}l}                
    C   & ::= & & \mbox{\bf\textsf{Concrete Qualifiers}} \\
        & | & \top \mbox{ or } \bot & \mbox{two-point lattice elements}
    \end{array}
  \]
  Lattice facts reminder: $\bot  \sqsubseteq \bot$, $\bot  \sqsubseteq \top$, and $\top  \sqsubseteq \top$.  $\top \sqcap C = C$,
  $\top \sqcup C = \top$, $\bot \sqcap C = \bot$, and $\bot \sqcup C = C$.
  \end{minipage}
  \end{center}
      \caption{The syntax of \Fq. Qualified differences to \Fsub highlighted in $\highlight{\textit{grey}}$.}
  \label{fig:base:syntax}

\end{figure}

We are now ready to present our recipe by constructing \Fq, a qualified extension of \Fsub with support for type qualifiers, polymorphism over type qualifiers, as well as meets ($Q \wedge R$) and joins ($Q \vee R$) over qualifiers.  We start by constructing
a simplified version of \Fq which models a free lattice over a two-point qualifier lattice to illustrate our recipe.

\paragraph*{Assigning Qualifiers} In \Fq we qualify types with the free lattice generated over a base two-point lattice with $\top$ and $\bot$, but provide no interpretation of $\top$ and $\bot$ as \Fq is only a base calculus.

\paragraph*{Syntax}
Figure~\ref{fig:base:syntax} presents the syntax of \Fq, with additions over \Fsub highlighted in grey. Type qualifiers $Q$ not only include $\top$ and $\bot$ as they would be in \citet{10.1145/301618.301665}'s
original system.  Here, in addition we support \emph{qualifier variables} $Y$, as well as meets and joins over qualifiers.  Type variables support polymorphism over {\it unqualified} types. To support qualifier polymorphism, we add a new qualifier for-all form $\Qfntype{Y}{Q}{T}$. Similarly, on the term-level we add qualifier abstraction $\Qfna{P}{Y}{Q}{t}$ and qualifier application $\Qapp{s}{Q}$.  

To ensure that qualifiers have some runtime semantics in our base calculus, we {\it tag} values with a qualifier expression $P$ denoting the qualifier that value should be typed at and we add support for {\it asserting} as well as {\it upcasting} qualifier tags, following \citet[Section 2.2]{10.1145/301618.301665}.
While \Fq does not provide a default tag for values, negative (or contravariant) qualifiers like \qual{noexcept}
would inform a default qualifier tag choice of $\top$ -- by default, functions are assumed to throw -- and positive (or covariant) qualifiers like \qual{const} would inform a default qualifier tag choice of $\bot$ -- 
by default, in mutable languages, values should be mutable.  Put simply, the default value tag should
correspond to the default, omitted, qualifier.

\begin{figure}
  \judgement{Evaluation for \Fq}{\fbox{$s \reduces t$} and \fbox{$\eval{Q}$}}

  \begin{minipage}{0.60\textwidth}{\small
    \vspace{1.1em}
    \infax[\ruledefN{beta-v}{beta-v}]{
       \app{(\fna{P}{x}{t})}{s} \reduces t[x \mapsto s]
    }
    \vspace{1em}
    \infax[\ruledefN{beta-T}{beta-T}]{
        \Sapp{(\Sfna{P}{X}{S}{t})}{S'} \reduces t[X \mapsto S']
    }
    \vspace{1em}
    \infax[\ruledefN{beta-Q}{beta-Q}]{
        \Qapp{(\Qfna{P}{Y}{Q}{t})}{Q'} \reduces t[X \mapsto Q']
    }

    \vspace{1em}
    \infrule[\ruledefN{upqual}{upqual}]
        {v \mbox{ tagged with } P \gap \eval(P)  \sqsubseteq \eval(Q)}
        {\upqual P~v \reduces v \mbox{ retagged with } Q }

    \vspace{1em}
    \infrule[\ruledefN{assert}{assert}]
        {v \mbox{ tagged with } P \gap \eval(P)  \sqsubseteq \eval(Q)}
        {\assert P~v \reduces v}
  }
  \end{minipage}\hfill
  \begin{minipage}{0.35\textwidth}{\small
    \vspace{1em}
    \infrule[\ruledefN{context}{context}]{
        s \reduces  t
    }{
        E[s]  \reduces E[t]
    }}
    \[
    \begin{array}{lcll}
        E & ::= & \mbox{\bf\textsf{Evaluation Context}} & \\
          &  |  & [] \\
          &  |  & E(t) ~|~ v(E) \\
          &  |  & E[S]~|~E[Q] \\
          &  |  & \upqual P~E \\
          &  |  & \assert P~E
    \end{array}
    \]
    \end{minipage}

    \begin{center}
    \[
      \begin{array}[t]{rlll}\\
       \eval(Q) & ::= & &\mbox{\bf\textsf{Partial Qualifier Evaluation}} \\
             & |~C & => & C \\
             & |~Q \wedge R & => & \eval(Q) \sqcap \eval(R) \\
             & |~Q \vee R & => & \eval(Q) \sqcup \eval(R) \\
             & |~\_ & => & \mbox{nothing, otherwise.}
      \end{array}
    \]
    \end{center}
    \caption{Reduction rules for \Fq}
    \label{fig:base:evaluation}
    \end{figure}

\paragraph*{Semantics}
The evaluation rules of \Fq (defined in Figure \ref{fig:base:evaluation}) are largely unchanged from \Fsub.  To support {\it qualifier polymorphism} we add the rule \ruleref{beta-Q} for reducing applications of a qualifier abstraction to a type qualifier expression.  Finally, to ensure that qualifiers have some runtime semantics even in our base calculus we add the rules \ruleref{upqual} and \ruleref{assert} for asserting and upcasting qualifier tags: they coerce qualifier expressions to concrete qualifiers when possible and ensure that the concrete qualifiers are compatible before successfully reducing.

\paragraph*{Subqualification}

\begin{figure}
  \judgement{Subqualification for \Fq}{\fbox{$\Gamma \vdash Q \sub R$}}

  \begin{minipage}{0.45\textwidth}
  \vspace{1.1em} 
  \infax[\ruledef{sq-top}]{%
    \Gamma \vdash Q \sub \top}
\vspace{1em}
  \infax[\ruledef{sq-bot}]{%
    \Gamma \vdash \bot \sub Q}

  \vspace{1em}
  \infrule[\ruledef{sq-join-intro-1}]{%
    \Gamma \vdash Q \sub R_1
  }{%
    \Gamma \vdash Q \sub R_1 \vee R_2
  }
  \infrule[\ruledef{sq-join-intro-2}]{%
    \Gamma \vdash Q \sub R_2
  }{%
    \Gamma \vdash Q \sub R_1 \vee R_2
  }
  \vspace{1em}
  \infrule[\ruledef{sq-join-elim}]{
    \Gamma \vdash R_1 \sub Q \andalso \Gamma \vdash R_2 \sub Q
  }{
    \Gamma \vdash R_1 \vee R_2 \sub Q
  }

  \end{minipage}\hfill
  \begin{minipage}{0.45\textwidth}
  \vspace{1.1em}

  \infrule[\ruledef{sq-meet-elim-1}]{%
    \Gamma \vdash R_1 \sub Q
  }{%
    \Gamma \vdash R_1 \wedge R_2 \sub Q
  }
  \infrule[\ruledef{sq-meet-elim-2}]{%
    \Gamma \vdash R_2 \sub Q
  }{%
    \Gamma \vdash R_1 \wedge R_2 \sub Q
  }
  \vspace{1em}
  \infrule[\ruledef{sq-meet-intro}]{
    \Gamma \vdash Q \sub R_1 \andalso \Gamma \vdash Q \sub R_2
  }{
    \Gamma \vdash Q \sub R_1 \wedge R_2
  }
    \vspace{1em}
    \infrule[\ruledef{sq-var}]{
        Y \sub Q \in \Gamma  \gap  \Gamma \vdash Q \sub R
    }{
        \Gamma \vdash Y \sub R
    }
    \vspace{1em}
    \infrule[\ruledef{sq-refl-var}]{
        Y \sub Q \in \Gamma
    }{
        \Gamma \vdash Y \sub Y
    }

  \end{minipage}

\caption{Subqualification rules of \Fq.}
\label{fig:base:subcapturing}

\end{figure}

Figure \ref{fig:base:subcapturing} captures the free lattice structure of the qualifiers of \Fq with a {\it subqualification} judgment $\Gamma \vdash Q \sub R$ to make
precise the partial order between two lattice formulas in a free lattice. 
This basic structure should appear familiar---it is a simplified subtyping lattice. 
It should not be surprising that this construction gives rise to the free lattice, though we make this property explicit in 
supplementary material.
One can use this structure to deduce desirable subqualification judgments; for an environment
$\Gamma = [X \sub A, Y \sub B, A \sub \top, B \sub \top]$, we can show that $X \vee Y \sub A \vee B$, using the following rule applications:
{\small
\begin{align*}
    X <: A \vee B & \text{ by } \ruleref{sq-join-intro-1} \\
    Y <: A \vee B & \text{ by } \ruleref{sq-join-intro-2} \\
    X \vee Y <: A \vee B & \text{ by } \ruleref{sq-join-elim}
\end{align*}
}
\paragraph*{Subtyping}

\begin{figure}
  \judgement{Subtyping for \Fq}{\fbox{$\Gamma \vdash S_1 \sub S_1$ and $\Gamma \vdash T_1 \sub T_2$}}

  \begin{minipage}{0.45\textwidth}
  \vspace{1.1em} 
  \infax[\ruledef{sub-top}]{%
    \Gamma \vdash S \sub \top}

  \vspace{1em}

  \infrule[\ruledef{sub-refl-svar}]{%
    X \in \Gamma
  }{%
    \Gamma \vdash X \sub X}

  \vspace{1em}

  \infrule[\ruledef{sub-svar}]{%
    X \sub S_1 \in \Gamma \gap \Gamma \vdash S_1 \sub S_2
  }{%
    \Gamma \vdash X \sub S_2}

    \vspace {1em}
  \infrule[\ruledef{sub-qtype}]{
    \Gamma \vdash Q_1 \sub Q_2 \gap \Gamma \vdash S_1 \sub S_2
  }{
    \Gamma \vdash \qtyp{Q_1}{S_1} \sub \qtyp{Q_2}{S_2}
  }
  \vspace{1em}

  \end{minipage}\hfill
  \begin{minipage}{0.45\textwidth}
  \vspace{1.1em}
  \infrule[\ruledef{sub-arrow}]{
    \Gamma \vdash T_1 \sub T_2 \gap \Gamma \vdash T_3 \sub T_4
  }{
    \Gamma \vdash \fntype{T_1}{T_3} \sub \fntype{T_2}{T_4}
  }

  \vspace{1em}
  \infrule[\ruledef{sub-all}]{
    \Gamma \vdash S_2 \sub S_1 \gap \Gamma, X \sub S_1 \vdash T_1 \sub T_2
  }{
    \Gamma \vdash \Sfntype{X}{S_1}{T_1} \sub \Sfntype{X}{S_2}{T_2}
  }
  \vspace{1em}

  \infrule[\ruledef{sub-qall}]{
    \Gamma \vdash Q_2 \sub Q_1 \gap \Gamma, Y \sub Q_1 \vdash T_1 \sub T_2
  }{
    \Gamma \vdash \Sfntype{Y}{Q_1}{T_1} \sub \Sfntype{Y}{Q_2}{T_2}
  }

  \end{minipage}

\caption{Subtyping rules of \Fq.}
\label{fig:base:subtyping}

\end{figure}

\Fq inherits most of its rules for subtyping from \Fsub, with two changes made (Figure \ref{fig:base:subtyping}). The additional rule \ruleref{sub-qall} handles
subtyping for qualifier abstractions, and rule \ruleref{sub-qtype}
handles subtyping for qualified types.
All other rules remain unchanged, except that rules \ruleref{sub-arrow}, \ruleref{sub-all}, and \ruleref{sub-qall} are updated to operate on qualified types $T$ (instead of simple types $S$).

\paragraph*{Typing}

\begin{figure}
  \judgement{Typing for \Fq}{\fbox{$\Gamma \vdash t : T$}}

  \begin{minipage}{0.45\textwidth}
  \vspace{1.1em} 
  \infrule[\ruledef{var}]{
    x : T \in \Gamma
  }{
    \Gamma \vdash x : T
  }
  \vspace{1em}
  \infrule[\ruledef{abs}]{
    \Gamma, x : T_1 \vdash t : T_2
  }{
    \Gamma \vdash \fna{P}{x}{t} : \qtyp{\highlight{P}}{\fntype{T_1}{T_2}}
  }
  
  \vspace{1em}
  \infrule[\ruledef{t-abs}]{
    \Gamma, X \sub S \vdash t : T
  }{
    \Gamma \vdash \Sfna{P}{X}{S}{t} : \qtyp{\highlight{P}}{\Sfntype{X}{S}{T}}
  }

  \vspace{1em}
  \infrule[\ruledef{q-abs}]{
    \Gamma, X \sub S \vdash t : T
  }{
    \Gamma \vdash \Qfna{P}{Y}{Q}{t} : \qtyp{\highlight{P}}{\Sfntype{Y}{Q}{T}}
  }
  \vspace{1em}
  \infrule[\ruledef{typ-assert}]{
    \Gamma \vdash t : \qtyp{Q}{S} \gap \Gamma \vdash Q \sub P
  }{
    \Gamma \vdash \assert P~t : \qtyp{Q}{S}
  }
  \end{minipage}\hfill
  \begin{minipage}{0.45\textwidth}

  \vspace{1.1em} 
  
  \infrule[\ruledef{app}]{
    \Gamma \vdash t : \qtyp{Q}{\fntype{T_1}{T_2}} \gap \Gamma \vdash s : T_1
  }{
    \Gamma \vdash \app{t}{s} : T_2
  }
  
  \vspace{1em}
  \infrule[\ruledef{t-app}]{
    \Gamma \vdash t : \qtyp{Q}{\Sfntype{X}{S}{T}} \gap \Gamma \vdash S' \sub S 
  }{
    \Gamma \vdash \Sapp{t}{S'} : T[X \mapsto S']
  }
  
  \vspace{1em}
  \infrule[\ruledef{q-app}]{
    \Gamma \vdash t : \qtyp{R}{\Sfntype{Y}{Q}{T}} \gap \Gamma \vdash Q' \sub Q
  }{
    \Gamma \vdash \Qapp{t}{Q'} : T[Y \mapsto Q']
  }

\vspace{1em}
  \infrule[\ruledef{sub}]{
    \Gamma \vdash s : T_1 \gap \Gamma \vdash T_1 \sub T_2
  }{
    \Gamma \vdash s : T_2
  }
  \vspace{1em}
  \infrule[\ruledef{typ-upqual}]{
    \Gamma \vdash t : \qtyp{Q}{S} \gap \Gamma \vdash Q \sub P
  }{
    \Gamma \vdash \upqual P~t : \qtyp{P}{S}
  }
  \end{minipage}
    \caption{Typing rules for \Fq}
    \label{fig:base:typing}

\end{figure}

Finally, Figure \ref{fig:base:typing} defines the typing rules of \Fq. The typing judgment assigns qualified types $T$ to expressions, and can be viewed as $\Gamma \vdash t : \qtyp{Q}{S}$.  As \Fq does not assign an interpretation to qualifiers, the introduction rules for typing values, \ruleref{abs}, \ruleref{t-abs}, and \ruleref{q-abs}, simply introduce qualifiers by typing values with their tagged qualifier, and the elimination rules remain unmodified. The only (new) elimination rules which deal with qualifiers are the new rules \ruleref{typ-assert} and \ruleref{typ-upqual}, which check that their argument is properly qualified.  We additionally add \ruleref{q-abs} and \ruleref{q-app} to support qualifier polymorphism. Besides these changes, the typing rules immediately carry over from \Fsub.

\subsection{Metatheory}
\label{section:metatheory}
\Fq satisfies the standard progress and preservation theorems.

\begin{theorem}[Preservation]
Suppose $\Gamma \vdash s : T$, and $s \reduces t$. Then $\Gamma \vdash t : T$ as well.
\end{theorem}
\begin{theorem}[Progress]
Suppose $\varnothing \vdash s : T$. Then either $s$ is a value, or $s \reduces t$ for some term $t$.
\end{theorem}

While \Fq does not place any interpretation on qualifiers outside of $\upqual$ and $\assert$,
such a system can already be useful.  For one, the static type of a value will always be greater than
the tag annotated on it and this correspondence is preserved through reduction by progress and preservation.
This property can already be used to enforce safety constraints.  For example, as \citet{10.1145/301618.301665} point out, one can use a negative type qualifier \qual{sorted} to distinguish between sorted and unsorted
lists.  By default most lists would be tagged at $\top$, marking them as unsorted lists.
A function like {\tt merge}, though, which merges two sorted lists into a third sorted list,
would expect two $\bot$-tagged lists, $\assert$ that they are actually $\bot$-tagged,
and produce a $\bot$-tagged list as well.  While this scheme does not ensure that all
$\bot$-tagged lists are sorted, so long as programmers are careful to ensure that they never
construct explicitly $\bot$-tagged {\it unsorted} lists, they can ensure that functions which expect
\qual{sorted} lists are actually passed sorted lists.

\subsection{Generalizing Qualifiers to General Lattices}
\label{section:multiple}
Qualifiers often come in more complicated lattices: for example, {\it protection rings} \cite{6234805}
induce a countable lattice, and combinations of binary qualifiers induce a product lattice.
Now, we show how we can tweak the recipe used to construct \Fq for two-point lattices
to support general (countable, bounded) qualifier lattices $\mathcal{L}$ as well.


\begin{figure}
   \[
  \begin{array}[t]{rll@{\hspace{4mm}}l}\\
   P, Q, R & ::= & & \mbox{{\bf\textsf{Qualifiers in extended \Fq}}} \\
        & |   & \highlight{l} & \mbox{Base lattice elements $l \in L$} \\
        & |   & Y & \mbox{Qualifier variables} \\
        & |   & Q \wedge R ~|~ Q \vee R & \mbox{Meets and joins} \\[\medskipamount]

    C   & ::= & & \mbox{\bf\textsf{Concrete Qualifiers}} \\
        & | & \highlight{l} & \mbox{Base lattice elements $l \in L$}
    \end{array}
    \]
  \caption{The syntax of \Fq extended over a bounded lattice $\mathcal{L}$. Differences to \Fq highlighted in $\highlight{\textit{grey}}$.}
  \label{fig:multiple:syntax}

\end{figure}

\begin{figure}
  \judgement{Subqualification for \Fq over a lattice $\mathcal{L}$}{\fbox{$\Gamma \vdash Q \sub R$}}

  \begin{minipage}{0.4\textwidth}
  \vspace{1.1em} 
  \infrule[\ruledef{sq-lift}]{%
    l_1, l_2 \in \mathcal{L} \gap l_1 \sqsubseteq l_2
  }{%
    \Gamma \vdash l_1 \sub l_2
  }

  \end{minipage}\hfill
  \begin{minipage}{0.5\textwidth}
  \vspace{1.1em}

  \infrule[\ruledef{sq-eval-elim}]{
    \Gamma \vdash Q \sub Q' \gap \Gamma \vdash l = \eval{Q'} \gap \Gamma \vdash l \sub R
  }{
    \Gamma \vdash Q \sub R
  }

  \vspace{1em}
  \infrule[\ruledef{sq-eval-intro}]{
    \Gamma \vdash Q \sub l \gap \Gamma \vdash l = \eval{Q'} \gap \Gamma \vdash Q' \sub R
  }{
    \Gamma \vdash Q \sub R
  }

  \end{minipage}

\caption{Extended sub-qualification rules for \Fq.}
\label{fig:multiple:subqualification}

\end{figure}

\paragraph*{Syntax}
The syntax changes needed to support this construction are listed in Figure \ref{fig:multiple:syntax}.  
Lattice elements are now generalized from $\top$ and $\bot$ to elements $l$ from our
base lattice $\mathcal{L}$, but as $\mathcal{L}$ is bounded, note that we still have distinguished elements $\top$ and $\bot$ in $\mathcal{L}$.

\paragraph*{Subqualification}
The subqualification changes needed to support this construction are listed in Figure \ref{fig:multiple:subqualification}.  These are exactly the rules needed to support the free lattice construction
over any arbritrary countable bounded lattice.
Rule \ruleref{sq-lift} simply lifts the lattice order $\sqsubseteq$ that $\mathcal{L}$ is equipped with up to the free lattice
order defined by the subqualification lattice. Rules \ruleref{sq-eval-elim} and \ruleref{sq-eval-intro} are a little more complicated, though, but are necessary in order to relate  {\it textual meets and joins} of elements of the base lattice $\mathcal{L}$, like $l_1 \vee l_2$, to their actual meets and joins in the qualifier lattice, $l_1 \sqcup l_2$.  We would expect that these
two terms would be equivalent in the subqualification lattice; namely, that $\Gamma \vdash l_1 \vee l_2 \sub l_1 \sqcup l_2$
and that $\Gamma \vdash l_1 \sqcup l_2 \sub l_1 \vee l_2$.
However, without the two evaluation rules \ruleref{sq-eval-elim} and \ruleref{sq-eval-intro} we would only
be able to conclude that $\Gamma \vdash l_1 \vee l_2 \sub l_1 \sqcup l_2$, but not the other desired inequality
$\Gamma \vdash l_1 \sqcup l_2 \sub l_1 \vee l_2$.  

To discharge this equivalence, \ruleref{sq-eval-elim} and \ruleref{sq-eval-intro} use $\eval$ to simplify
qualifier expressions. Again, it should not be surprising that this gives rise to the free lattice of extensions of
$\mathcal{L}$, though we make this precise in supplementary material. 

\paragraph*{Soundness}
Like simple \Fq, \Fq extended over an bounded lattice $\mathcal{L}$
also satisfies the standard soundness theorems:

\begin{theorem}[Preservation for Extended \Fq]
Suppose $\Gamma \vdash s : T$, and $s \reduces t$.  Then $\Gamma \vdash t : T$ as well.
\end{theorem}
\begin{theorem}[Progress for Extended \Fq]
Suppose $\varnothing \vdash s : T$.  Either $s$ is a value, or $s \reduces t$ for some term $t$.
\end{theorem}

\noindent
However this construction while sound poses some difficulties.
The subqualification rules now need to handle transitivity through
base lattice elements, and these new rules are not syntax directed.  It remains an open question as to whether or not extended \Fq admits algorithmic subtyping rules, and we suspect the answer depends on the structure of the base bounded qualifier lattice $\mathcal{L}$ being extended.

\section{Applications}
\label{section:applications}
Having introduced our design recipe by constructing \Fq as a qualified extension of \Fsub, we now study how our subqualification and polymorphism recipe can be reused in three practical qualifier systems.  For brevity we will base our qualifier systems
on \Fq as it already provides rules and semantics for typing, subqualification and qualifier polymorphism, which we modify below.

\subsection{Reference Immutability}
We start by examining one well-studied qualifier system, that of {\it reference immutability}  \cite{10.1145/1103845.1094828,10.1145/2384616.2384680}.  
In this setting, each (heap) reference can be either mutable or immutable.
An immutable reference cannot be used to mutate the value or any other
values transitively reached from it, so a value read through a \qual{readonly}-qualified compound object
or reference is itself \qual{readonly} as well. Mutable and immutable references can coexist for the same value, so an immutable reference does not itself guarantee that the value will not change through some other, mutable reference. This is in contrast to the stronger guarantee of {\it object immutability}, which applies to values, and ensures that a particular value does not change through any of the references to it~ \cite{10.1145/1287624.1287637}.

Reference immutability systems have long been studied in various contexts \cite{10.1145/1103845.1094828,10.1145/2384616.2384680,10.1145/1287624.1287637,10.1145/2384616.2384619,Lee2023ToAppear,dort_et_al:LIPIcs:2020:13175}.  Here, we show that we can reuse our recipe to model reference immutability in a setting with higher rank polymorphism and subtyping over both qualifiers and ground types, in a calculus \Fm. 

\paragraph*{Assigning Qualifiers}  We need to define how qualifiers \qual{mutable} and \qual{readonly} are assigned to $\top$ and $\bot$ in  \Fm. Since a \qual{mutable} reference can always be used where a \qual{readonly} reference is expected, we assign \qual{mutable} to $\bot$ and \qual{readonly} to $\top$.  This is reflected in Figure \ref{fig:mut:syntax}.

\paragraph*{Syntax and Evaluation}
Now we need to design syntax and reduction rules for references and immutable references.  We add support for references via $\Boxed$ forms and we add rules for introducing and eliminating boxes.  To distinguish between mutable and immutable boxes, we
reuse the qualifiers tagged on values--values with tags $P$ that $\eval$ to $\bot$ are mutable, whereas values with tags $P$ that otherwise evaluate to $\top$ are {\it mutable}.  One can explicitly mark a value immutable by $\upqual$-ing to $\top$.    The elimination form for reading from a reference, \ruleref{deref}, ensures
that a value read from a reference tagged \qual{immutable}, or at $\top$, remains \qual{immutable}.  This is reflected in the updated operational semantics (Figure \ref{fig:mut:evaluation}).  Reduction now takes place over pairs of terms and stores $\langle t, \sigma \rangle$; stores map locations $l$ to values.

\begin{figure}
  \begin{minipage}{0.47\textwidth}
  \[
  \begin{array}[t]{rll@{\hspace{4mm}}l}\\
   s, t & ::= & & \mbox{\bf\textsf{Terms}} \\
        & \hdots \\
        & | & \Boxed_P t & \mbox{reference cell} \\
        & | & \unboxed s & \mbox{deferencing} \\
        & | & \setboxed{s}{t} & \mbox{reference update}
  \end{array}
  \]
  \end{minipage}\hfill
  \begin{minipage}{0.47\textwidth}
  \[
  \begin{array}[t]{rll@{\hspace{4mm}}l}\\
   S    & ::= & & \mbox{{\bf\textsf{Types}}} \\
          & \hdots \\
          & | & \Boxed S& \mbox{reference type} \\[\medskipamount]
          
  P, Q, R  & ::= & & \mbox{\bf\textsf{Qualifiers}} \\
        & \hdots && \mbox{as before, except:} \\
        & |   & \readonly & \mbox{const qualifier (as $\top$)} \\
        & |   & \mutable & \mbox{non-const qualifier (as $\bot$)} \\[\medskipamount]
  \end{array}
  \]
  \end{minipage}
  \begin{minipage}{0.47\textwidth}
  \[
    \begin{array}[t]{rll@{\hspace{4mm}}l}
           &  & l & \mbox{\bf\textsf{Location}} \\[\medskipamount]
     s, t & ::= & & \mbox{\bf\textsf{Runtime Terms}} \\
          & |   & \Boxed_P l & \mbox{runtime reference} \\[\medskipamount]
                
    v   & ::= & & \mbox{\bf\textsf{Runtime Values}} \\
        & \hdots \\
        & | & \Boxed_P l
    \end{array}
  \]
  \end{minipage}\hfill
  \begin{minipage}[c]{0.47\textwidth}
  \[
    \begin{array}[t]{rll@{\hspace{4mm}}l}
     \sigma & ::= & & \mbox{\bf\textsf{Store}} \\
            & |   & \cdot & \mbox{empty} \\
            & |   & \sigma,~l : v & \mbox{cell $l$ with value $v$} \\[\medskipamount]

     \Sigma & ::= & & \mbox{\bf\textsf{Store Environment}} \\
            & |   & \cdot & \mbox{empty} \\
            & |   & \sigma,~l : T & \mbox{cell binding}
    \end{array}
  \]
  \end{minipage}
  \caption{The syntax of \Fm.}
  \label{fig:mut:syntax}

\end{figure}

\begin{figure}
  \judgement{Additional Evaluation Rules for \Fm}{\fbox{$\langle s, \sigma\rangle \reduces \langle t, \sigma' \rangle$}}

  \begin{minipage}{0.45\textwidth}{\small
    \vspace{1.1em}
    \infrule[\ruledefN{ref-store}{ref-store}]{
        l \notin \sigma
    }{
        \langle \Boxed_P v, \sigma \rangle \reduces \langle \Boxed_P l, (\sigma, l : v) \rangle 
    }
    \vspace{1em}
    \infrule[\ruledefN{deref}{deref}]{
        l : v \in \sigma \gap v \mbox{ tagged with } Q
    }{
        \langle \unboxed \Boxed_P l, \sigma \rangle \reduces  \langle v \mbox{ retagged at } P \vee Q, \sigma \rangle
    }
  }
    
  \end{minipage}\hfill
  \begin{minipage}{0.45\textwidth}{\small
    \vspace{1em}
    \infrule[\ruledefN{write-ref}{write-ref}]{
        l : v \in \sigma \gap \eval(P) \sqsubseteq \bot
    }{
        \langle \setboxed (\Boxed_P~l~)~v', \sigma \rangle \mapsto \langle v, \sigma[l \mapsto v'] \rangle
    }
    \vspace{1em}
    \infrule[\ruledefN{context}{context}]{
        \langle s, \sigma \rangle \reduces \langle t, \sigma' \rangle
    }{
        \langle E[s], \sigma \rangle \reduces \langle E[t], \sigma' \rangle
    }}
    \end{minipage}

    \begin{center}
    \[
    \begin{array}{lcll}
        E & ::= & \hdots & \mbox{{\bf Evaluation Context}}\\
            & |   & \Boxed_P E \\
            & |   & \unboxed E \\
            & |   & \setboxed E~t~|~\setboxed v~E    \end{array}
    \]
    \end{center}

    \caption{Reduction rules for \Fm}
    \label{fig:mut:evaluation}

    \end{figure}

\paragraph*{Typing}

\begin{figure}
  \judgement{Additional Typing and Runtime Typing for \Fm}{\fbox{$\Gamma~|~\Sigma \vdash t : T$ and $\Gamma~|~\highlight{\Sigma} \vdash \sigma$}}

  \begin{minipage}{0.5\textwidth}
  \vspace{1.1em} 
  \infrule[\ruledef{ref-intro}]{
    \Gamma~|~\Sigma \vdash t : T
  }{
    \Gamma~|~\Sigma \vdash \Boxed_{P} t : \qtyp{P}{\Boxed T}
  }
  \vspace{1em}
  \infrule[\ruledef{runtime-ref-intro}]{
    l : T \in \Sigma
  }{
    \Gamma~|~\Sigma \vdash \Boxed_{P} l : \qtyp{P}{\Boxed T}
  }
  \end{minipage}\hfill
  \begin{minipage}{0.5\textwidth}

  \vspace{1.1em} 

  \vspace{1em}
  \infrule[\ruledef{ref-elim}]{
    \Gamma~|~\Sigma \vdash t : \qtyp{Q_1}{\Boxed \qtyp{Q_2}{S}}
  }{
    \Gamma~|~\Sigma \vdash \unboxed t : \qtyp{\highlight{Q_1 \vee Q_2}}{S}
  }
  
  \vspace{1em}
  \infrule[\ruledef{ref-update}]{
    \Gamma \vdash s : \qtyp{\highlight{\mutable}}{\Boxed T} \gap \Gamma \vdash t : T
  }{
    \Gamma \vdash \setboxed s~t : T
  }
  \end{minipage}

    \begin{center}
      \begin{minipage}{0.75\textwidth}      
       \vspace{1em}

      \infrule[\ruledef{store}]{
        dom(\sigma) = dom(\Sigma) \gap \forall l \in dom(\Sigma),~\Gamma~|~\Sigma \vdash \sigma(l) : \Sigma(l)
      }{
        \Gamma~|~\Sigma \vdash \sigma
      }
      \end{minipage}
    \end{center}
    \caption{Typing rules for \Fm; notable changes highlighted in {\it grey}.}
    \label{fig:mut:typing}

\end{figure}

 We now need to define new typing rules for reference forms and to possibly adjust existing typing rules to account for our new runtime interpretation of qualifiers.  For this system, we only need to add typing rules, as shown in Figure~\ref{fig:mut:typing}. To ensure immutability safety, the standard reference update elimination form
 \ruleref{ref-update} is augmented to check that a reference can only be written to if and only if it can be typed as \qual{mutable} $\Boxed{}$. Finally, the standard reference read elimination form \ruleref{ref-elim} is augmented to enforce that the mutability of the value read from a reference is joined with the mutability of the reference itself to ensure transitive immutability safety.  Other than qualifiers, our construction is completely standard; we merely add a store $\sigma$ and a runtime store environment $\Sigma$ mapping store locations to types.  

\paragraph*{Metatheory}
We can prove the standard soundness theorems without any special difficulty:
\begin{theorem}[Preservation of \Fm]
    \label{theorem:mut:preservation}
    Suppose $\langle s, \sigma \rangle \reduces \langle t, \sigma' \rangle$. 
    If $\Gamma~|~\Sigma \vdash \sigma$ and $\Gamma~|~\Sigma \vdash s : T$ for some type $T$,
    then there is some environment extension $\Sigma'$ of $\Sigma$ such that $\Gamma~|~\Sigma' \vdash \sigma'$ and
    $\Gamma~|~\Sigma' \vdash t : T$.
\end{theorem}
\begin{theorem}[Progress for \Fm]
    \label{theorem:mut:progress}
    Suppose $\varnothing~|~\Sigma \vdash \sigma$ and $\varnothing, \Sigma \vdash s : T$.  Then either
    $s$ is a value or there is some $t$ and $\sigma'$ such that $\langle s, \sigma \rangle \reduces \langle t, \sigma' \rangle$.
\end{theorem}

\noindent With only progress and preservation, we can already state something meaningful about the immutability safety of \Fm: we know
that well-typed programs will not get stuck trying to write to a sealed,$\bot$-tagged reference.  Moreover, the typing rules, in
particular \ruleref{ref-elim}, give us our desired transitive immutability safety as well; values read from a
$\bot$-tagged value will remain $\bot$-tagged and therefore immutable as well.  In addition, as qualifier tags
only affect reduction by blocking reduction (that is, getting stuck) we almost directly recover full immutability
safety as well for free, by noting that references typed (by subtyping) at \qual{readonly} can be re-tagged
at \qual{readonly} as well without affecting reduction, assuming the original program was well-typed.

\subsection{Function Colouring}
Function colouring \cite{Nystrom_2015} is another qualifier system. In this setting, functions are qualified
with a kind that indicates a {\it colour} for each function, and there are restrictions on which other functions a function
can call depending on the colours of the callee and caller.  For example, \qual{noexcept} and \qual{throws} forms a function colouring system---functions qualified \qual{noexcept} can only call functions qualified \qual{noexcept}.  Another instantiation of this problem
is the use of the qualifiers \qual{sync} and \qual{async} in asynchronous programming. 
\qual{async}-qualified functions may call all functions but \qual{sync}-qualified functions may only call other \qual{sync}-qualified functions.
Polymorphism with function colours is known to be painful \cite{Nystrom_2015}.  Consider a higher-order function {\tt map}:
\begin{lstlisting}[language=Scala]
    def map[X, Y](l: List[X], f: (X => Y)) = ???
\end{lstlisting}
What should its colour be? The colour of a function like {\tt map} depends on the function {\tt f} it is applying. Without a mechanism to express this dependency, such as colour polymorphism, functions like {\tt map} need to be implemented twice---once for an \qual{async}-qualified {\tt f}, and once for a \qual{sync}-qualified {\tt f}.
Moreover, function colouring requires a mechanism for mixing colours!  Consider function composition:
\begin{lstlisting}[language=Scala]
    def compose[A, B, C, D](f: A => B, g: C => D) = (x) => g(f(x))  
\end{lstlisting}
The colour of the result of {\tt compose} needs to be the {\it join} of the colours of {\tt f} and {\tt g}.  If either
{\tt f} or {\tt g} are asynchronous then the result of {\tt compose} is as well, but if both {\tt f} and {\tt g}
are synchronous then so should the result of composing them. 
We now show how our recipe can be used to construct \Fa, a calculus that enforces these restrictions.

\paragraph*{Assigning Qualifiers} Since a synchronous function can be called anywhere that an asynchronous function could be, we assign the $\top$ qualifier  to \qual{async} and the $\bot$ qualifier to \qual{sync}.

\paragraph*{Syntax}

\begin{figure}
  \begin{minipage}{0.5\textwidth}
  \[
  \begin{array}[t]{rll@{\hspace{4mm}}l}\\
  P, Q, R  & ::= & & \mbox{\bf\textsf{Qualifiers}} \\
        & \hdots && \mbox{as before, except:} \\
        & |   & \async ~(\mbox{as } \top) & \mbox{async qualifier} \\
        & |   & \sync ~(\mbox{as } \bot) & \mbox{sync qualifier}\\[\medskipamount]
        
    \kappa & ::= & & \mbox{\bf\textsf{Evaluation Context}} \\
           & | & [] \\ 
           & | & f :: \kappa \\[\medskipamount]
  \end{array}
  \]
  \end{minipage}
  \begin{minipage}{0.45\textwidth}
  \[
  \begin{array}[t]{rll@{\hspace{4mm}}l}\\
    f & ::= & & \mbox{\bf\textsf{Evaluation Frames}} \\
           & | & \highlight{\kbarrier C} & \mbox{barrier} \\
           & | & \karg t & \mbox{argument} \\
           & | & \kapp v & \mbox{application} \\
           & | & \ktarg T & \mbox{type application} \\
           & | & \kqarg Q & \mbox{qualifier application} \\[\medskipamount]


  \end{array}
  \]
  \end{minipage}
  \caption{The syntax of \Fa.}
  \label{fig:async:syntax}

\end{figure}

Figure~\ref{fig:async:syntax} presents the modified syntax of \Fa.
To keep track of the synchronicity a function
term should run in we reuse the tags already preset in values. An example of an asynchronous function term
is $\fna{\tt{async}}{x}{\;x}$, and an example of a function that is polymorphic in its qualifier is 
$\Qfna{\tt{async}}{Y}{\tt{sync}}{\fna{Y}{f}{\;f(1)}}$, describing a function that should run in the same synchronicity context as its argument $f$.

\paragraph*{Evaluation}

\begin{figure}
  \judgement{Evaluation for \Fa}{\fbox{$\cks{c}{\kappa} \reduces \cks{c'}{\kappa'}$}}

  \begin{minipage}{0.45\textwidth}{\footnotesize
    \vspace{1.1em}
    \infax[\ruledefN{cong-app}{cong-app}]{
       \cks{\app{s}{t}}{\kappa} \reduces \cks{s}{\karg t :: \kappa}
    }
    \vspace{1em}
    \infax[\ruledefN{cong-arg}{cong-arg}]{
        \cks{v}{\kapp t :: \kappa} \reduces \cks{t}{\kapp v :: \kappa}
    }
    \vspace{1em}
    \infax[\ruledefN{cong-tapp}{cong-tapp}]{
       \cks{\Sapp{s}{S}}{\kappa} \reduces \cks{s}{\ktarg S :: \kappa}
    }
    \vspace{1em}
    \infax[\ruledefN{cong-qapp}{cong-qapp}]{
       \cks{\Qapp{s}{Q}}{\kappa} \reduces \cks{s}{\kqarg Q :: \kappa}
    }
    \vspace{1em{}}
    \infax[\ruledefN{break-barrier}{break-barrier}]{
        \cks{v}{\kbarrier C :: \kappa} \reduces \cks{v}{\kappa}
    }
    \vspace{1em}}
  \end{minipage}\hfill
  \begin{minipage}{0.55\textwidth}{\footnotesize
    \vspace{1em}
    \infrule[\ruledefN{reduce-app}{reduce-app}]{        
        C \leq C_i \text{ for all } \kbarrier~C_i \text{ frames on } \kappa \gap \eval{P} = C
    }{
        \cks{v}{\kapp \fna{P}{x}{t} :: \kappa} \reduces \cks{t[x \mapsto v]}{\kbarrier C :: \kappa}
    }
    \vspace{1em}
    \infrule[\ruledefN{reduce-tapp}{reduce-tapp}]{        
        C \leq C_i \text{ for all } \kbarrier~C_i \text{ frames on } \kappa \gap \eval{P} = C
    }{
        \cks{\Sfna{P}{X}{S}{t}}{\ktarg S' :: \kappa} \reduces \cks{t[X \mapsto S']}{\kbarrier C :: \kappa}
    }
    \vspace{1em}
    \infrule[\ruledefN{reduce-qapp}{reduce-qapp}]{        
         C \leq C_i \text{ for all } \kbarrier~C_i \text{ frames on } \kappa\gap \eval{P} = C
    }{
    \cks{\Qfna{P}{Y}{Q}{t}}{\kqarg Q' :: \kappa} \reduces \cks{t[Y \mapsto Q']}{\kbarrier C :: \kappa}
    }
    }
    \end{minipage}

    \caption{Operational Semantics (CK-style) for \Fa}
    \label{fig:async:evaluation}

    \end{figure}

To model synchronicity safety, Figure~\ref{fig:async:evaluation} describes the operational semantics of \Fa using \citet{10.1145/41625.41654}-style CK semantics, extended with special {\it barrier} frames installed on the stack denoting the colour of the function that was called.  When a function is called,
we  place a {\it barrier} with the evaluated colour of the function itself, and functions
may only be called if the barriers on the stack are {\it compatible} with the evaluated colour of the function being called---namely, an asynchronous function can be called only if there are no barriers on the stack marked synchronous.  The other evaluation contexts are standard.

\paragraph*{Typing}
To guarantee soundness, Figure~\ref{fig:async:typing} endows the typing rules of \Fa with modified rules for keeping track of the synchronicity context
that a function needs.  We extend the typing rules with a colour context $R$ to keep track of the synchronicity of the functions being called.  This colour context $R$ is simply a qualifier expression, and is introduced
by the introduction rules for typing abstractions by lifting the qualifier tagged on those abstractions -- see  rules \ruleref{A-abs}, \ruleref{A-t-abs}, and \ruleref{A-q-abs}.  To ensure safety when applying functions in the elimination \ruleref{A-app}, we check that the colour context is compatible with the type of the function being called; subsumption in \ruleref{A-sub-eff} allows functions to run if the
qualifiers do not exactly match but when the qualifier on the function is subqualified by the colour context.  The typing rules outside of manipulating the context $R$ remain otherwise unchanged.

\paragraph*{Metatheory}
With all this, we can state and prove progress and preservation for \Fa.
\begin{theorem}[Progress of \Fa]
Suppose $\langle c, \kappa \rangle$ is a well-typed machine configuration.  Then either $c$ is a value and $k$
is the empty continuation, or there is a machine state $\langle c', \kappa'\rangle$ that it steps to.
\end{theorem}

\begin{theorem}[Preservation of \Fa]
Suppose $\langle c, \kappa \rangle$ is a well-typed machine configuration.  Then if it steps to
another configuration $\langle c', \kappa' \rangle$, that configuration is also well typed.
\end{theorem}

\begin{figure}
  \judgement{Typing for \Fa}{\fbox{$\Gamma \highlight{~|~ R} \vdash s : T$}}

  \begin{minipage}{0.5\textwidth}
  \small
  \vspace{1.1em} 
  \infrule[\ruledef{A-var}]{
    x : T \in \Gamma
  }{
    \Gamma~|~R \vdash x : T
  }
  \vspace{1em}
  \infrule[\ruledef{A-abs}]{
    \Gamma, x : T_1\highlight{~|~P} \vdash t : T_2
  }{
    \Gamma~\highlight{~|~\sync} \vdash \fna{P}{x}{t} : \qtyp{P}{\fntype{T_1}{T_2}}
  }
  
  \vspace{1em}
  \infrule[\ruledef{A-t-abs}]{
    \Gamma, X \sub S\highlight{~|~P} \vdash t : T
  }{
    \Gamma\highlight{~|~\sync} \vdash \Sfna{P}{X}{S}{t} : \qtyp{P}{\Sfntype{X}{S}{T}}
  }

  \vspace{1em}
  \infrule[\ruledef{A-q-abs}]{
    \Gamma, Y \sub Q \highlight{~|~P} \vdash t : T
  }{
    \Gamma\highlight{~|~\sync} \vdash \Qfna{P}{Y}{Q}{t} : \qtyp{P}{\Sfntype{Y}{Q}{T}}
  }

  \end{minipage}\hfill
  \begin{minipage}{0.45\textwidth}
  \small
  \vspace{1.1em} 
  
  \infrule[\ruledef{A-app}]{
    \Gamma\highlight{~|~R} \vdash t : \qtyp{\highlight{R}}{\fntype{T_1}{T_2}} \gap \Gamma \vdash s : T_1
  }{
    \Gamma\highlight{~|~R} \vdash \app{t}{s} : T_2
  }
  
  \vspace{1em}
  \infrule[\ruledef{A-t-app}]{
    \Gamma\highlight{~|~R} \vdash t : \qtyp{\highlight{R}}{\Sfntype{X}{S}{T}} \gap \Gamma \vdash S' \sub S 
  }{
    \Gamma\highlight{~|~R} \vdash \Sapp{t}{S'} : T[X \mapsto S']
  }
  
  \vspace{1em}
  \infrule[\ruledef{A-q-app}]{
    \Gamma\highlight{~|~R} \vdash t : \qtyp{\highlight{R}}{\Sfntype{Y}{Q}{T}} \gap \Gamma \vdash Q' \sub Q
  }{
    \Gamma\highlight{~|~R} \vdash \Qapp{t}{Q'} : T[Y \mapsto Q']
  }

\vspace{1em}
  \infrule[\ruledef{A-sub}]{
    \Gamma~|~R \vdash s : T_1 \gap \Gamma \vdash T_1 \sub T_2
  }{
    \Gamma~|~R \vdash s : T_2
  }

\vspace{1em}
  \infrule[\ruledef{A-sub-eff}]{
    \Gamma~|~R \vdash s : T_1 \gap \Gamma \vdash R \sub Q
  }{
    \Gamma~|~Q \vdash s : T_2
  }
  \end{minipage}
    \caption{Typing rules for \Fa}
    \label{fig:async:typing}

\end{figure}

\noindent Note that progress and preservation guarantee meaningful safety properties about \Fa, namely that
an asynchronous function is never called above a synchronous function during evaluation,
as such a call would get stuck, by \ruleref{reduce-app}.

\paragraph*{Observations} \Fa can be used to model function colouring with other qualifiers as well;
for example, we could model colours \qual{noexcept} and \qual{throws} by assigning \qual{noexcept}
to $\bot$ and \qual{throws} to $\top$.  More interestingly \Fa could be viewed as a simple {\it effect system}; the synchronicity context $R$ can be seen as the effect of a term!  We discuss this curious connection between qualifiers and effects in Section \ref{section:related:effects}.

\subsection{Tracking Capture}
Finally, our design recipe can be remixed to construct a qualifier system 
to qualify values based on what they capture. Some base values are {\it meaningful} and should be \qual{tracked}, 
and other values are {\it forgettable}. 

\paragraph*{Motivation}
One application of such a system is the {\it effects-as-capabilities} discipline \cite{10.1145/365230.365252}, which enables reasoning about which code can perform side effects by simply tracking capabilities, special values that grant the holder the ability to perform side effects; for example, the ability to perform I/O, or the ability to throw an exception.

\paragraph*{What to track?} Suppose for example we have a base capability named {\tt one\_ring}, which allows its holder to produce arbitrary values.  Such a precious value really ought to be \qual{tracked} and not forgotten, as in the hands of the wrong user, it can perform dangerous side effects!
\begin{lstlisting}[language=Scala]
    val one_ring : {tracked} [A] (Unit => A) = ???
\end{lstlisting}

\noindent However, it is not only {\tt one\_ring} itself that is dangerous.
Actors that {\it capture} {\tt one\_ring} can themselves cause dangerous side effects.
For example:
\begin{lstlisting}[language=Scala]
    def fifty_fifty(): Unit = {
        val gauntlet = one_ring[InfinityGauntlet]()
        gauntlet.snap()
    } // one_ring is captured by fifty_fifty.
\end{lstlisting}
In general, values that capture {\it meaningful} values---capabilities---become {\it meaningful} themselves, since they can perform side effects, so they should also be \qual{tracked}. 
Now, while it is clear that {\tt one\_ring} and {\tt fifty\_fifty} are both dangerous, they are dangerous for different reasons: {\tt one\_ring} because it intrinsically is and {\tt fifty\_fifty} because it captures {\tt one\_ring}. 

\paragraph*{Distinguishing Capabilities}
In practical applications, we may wish to distinguish between different effects, modelled by
different capabilities. For example, we may wish to reason about a more pedestrian side effect -- printing --
separately from the great evil that {\tt one\_ring} can perform.  It is reasonable to expect
that we can {\tt print} in more contexts than we can use the {\tt one\_ring}.

\begin{lstlisting}[language=Scala]
    val print : {tracked} String => Unit = ???
    def hello_world() = print "Hello World!"  // tracked as it captures print
    def runCodeThatCanPrint(f: ??? () => Unit) = f()
    runCodeThatCanPrint(hello_world) // OK
    runCodeThatCanPrint(fifty_fifty) // Should be forbidden
\end{lstlisting}
\noindent In this example, function {\tt runCodeThatCanPrint} only accepts thunks that print as a side effect.  What type annotation should we give to its argument {\tt f}? 
In particular, what qualifier should we use to fill in the blank? 
It should not be \qual{tracked}, as otherwise we could pass {\tt fifty\_fifty} to {\tt runCodeThatCanPrint} -- an operation which should be disallowed.  
Instead we would like to fill that blank with \qual{print}; to denote that {\tt runCodeThatCanPrint} can accept any thunk which is no more dangerous than {\tt print} itself. 
Figure~\ref{figure:qual:capture:example} summarizes the different variables in the above examples and the qualifiers we would like to assign to their types.

\begin{figure}
    \begin{center}
    \begin{tabular}{m{3cm} m{3cm} m{6cm}}
     {\bfseries\sffamily Term} & {\bfseries\sffamily Qualifier} & {\bfseries\sffamily Reason} \\ \toprule
     {\tt one\_ring} & \qual{tracked} & As {\tt one\_ring} is a base capability. \\ \midrule
     {\tt print} & \qual{tracked} & As {\tt print} is a base capability. \\ \midrule
     {\tt fifty\_fifty} & \qual{\bfseries one\_ring} & As {\tt fifty\_fifty} is no more dangerous than {\tt one\_ring}. \\ \midrule
     {\tt hello\_world} & \qual{\bfseries print} & As {\tt hello\_world} is no more
     dangerous than {\tt print}. \\
    \end{tabular}
    \end{center}        
    \caption{Qualifier assignments in Capture Tracking}
    \label{figure:qual:capture:example}
\end{figure}

As \citet{10.1145/3486610.3486893,boruchgruszecki2021tracking, 10.1145/3618003} show, such a capture tracking system could be used to guarantee desirable and important safety invariants.
They model capture tracking using sets of variables, but a set is just a lattice join of the singletons in that set!  For example, \citet{10.1145/3618003} would give the following {\tt evil\_monologue} function the capture set annotation {\tt \{fifty\_fifty, print\}}, while we would
give it the qualifier annotation {\tt \qual{\{fifty\_fifty | print\}}}.
\begin{lstlisting}[language=Scala]
    def evil_monologue(): Unit = {
        print "I expect you to die, Mr. Bond."
        fifty_fifty()
    } 
\end{lstlisting}

\noindent Using this insight, we can model capture tracking as an extension \Fc of \Fq.

\begin{figure}
  \begin{minipage}{0.45\textwidth}
  \[
  \begin{array}[t]{rll@{\hspace{4mm}}l}\\
   s, t & ::= & & \mbox{\bf\textsf{Terms}} \\
        & \hdots \\
        & | & \capp{s}{Q}{t} & \mbox{term application} \\[\medskipamount]
     S  & ::= & & \mbox{\bf\textsf{Types}} \\
        & \hdots \\
        & | & \cfntype{x}{T_1}{T_2} & \mbox{function type} \\[\medskipamount]
  P, Q, R  & ::= & & \mbox{\bf\textsf{Qualifiers}} \\
        & \hdots && \mbox{as before, except:} \\
        & |   & x & \mbox{term variables} \\
        & |   & \tracked ~(\mbox{as } \top) & \mbox{tracked values} \\
        

  \end{array}
  \]
  \end{minipage}
  \begin{minipage}{0.45\textwidth}
      \judgement{}{\bf\textsf{Evaluation: }\fbox{$s \reduces t$}}
      {\small
        \infax[\ruledefN{C-beta-v}{C-beta-v}]{
           \capp{(\fna{P}{x}{t})}{\highlight{Q}}{s} \reduces \\ t[x \mapstoty Q][x \mapstote s]
        }}

      \judgement{}{\bf\textsf{Subqualification: }\fbox{$\Gamma \vdash Q \sub R$}}{
      \small
        \infrule[\ruledef{sq-tvar}]{
            x : \qtyp{Q}{S} \in \Gamma  \gap  \Gamma \vdash Q \sub R
        }{
            \Gamma \vdash x \sub R
        }
        \infrule[\ruledef{sq-refl-tvar}]{
            x : \qtyp{Q}{S} \in \Gamma
        }{
            \Gamma \vdash x \sub x
        }
      }
      
  \end{minipage}

  \begin{center}
     \begin{minipage}{0.45\textwidth}
      \judgement{{\bf\textsf{Subtyping: }}}\fbox{$\Gamma \vdash S_1 \sub S_2$}{
      \small
          \infrule[\ruledef{C-sub-arrow}]{
            \Gamma \vdash T_1 \sub T_2 \gap \Gamma, x : T_1 \vdash T_3 \sub T_4
          }{
            \Gamma \vdash \cfntype{x}{T_2}{T_3} \sub \cfntype{x}{T_1}{T_4}
          }
          }
      \end{minipage}
  \end{center}

  \caption{Evaluation, Syntax, Subtyping for \Fc}
  \label{fig:capt:comb}
\end{figure}

\paragraph*{Assigning Qualifiers} 
We attach a qualifier \qual{tracked} to types, denoting which values we should keep track of. 
The qualifier \qual{tracked} induces a two-point lattice, where \qual{tracked} is at $\top$, and values that should
not be tracked, or should be {\it forgotten}, are qualified at $\bot$. 
Base capabilities will be given the \qual{tracked} qualifier.

\paragraph*{Syntax -- Tracking Variables}  
Figure~\ref{fig:capt:comb} defines the syntax of \Fc.  
'To reflect the underlying term-variable-based nature of capture tracking, term bindings in \Fc introduce both a term variable in term position as well as a qualifier variable in qualifier position with the same name as the term variable.

Term bindings now serve double duty introducing both term variables and qualifier variables, so a term like the monomorphic identity function $\fna{\bot}{x}{x}$
would be given the type $\qtyp{\bot}{\cfntype{x}{\qtyp{Q}{S}}{\qtyp{x}{S}}}$ to indicate that it is not tracked but the result might be tracked depending on whether or not its argument $x$ is tracked as well.  
This still induces a {\it free lattice} structure generated over the two-point lattice that \qual{tracked} induces, except in this case, the free lattice includes both
qualifier variables introduced by qualifier binders in addition to qualifier variables introduced by term binders as well. 
As term binders introduce both a term and qualifier variable, term application in \Fc now requires a qualifier argument
to be substituted for that variable in qualifier position.  
As such, term application in \Fc now has three arguments $\capp{s}{Q}{t}$ -- a function $s$, a qualifier $Q$, and an argument $t$; see Figure \ref{fig:capt:comb}.
In this sense, term abstractions in \Fc can be viewed as a combination of a qualifier abstraction $\Lambda[x <: Q]$ followed by a term abstraction $\lambda(x : \qtyp{x}{T})$.

\begin{figure}
  \judgement{Typing for \Fc}{\fbox{$\Gamma \vdash t : T$}}

  \begin{minipage}{0.5\textwidth}
  \vspace{1.1em} 

      \infrule[\ruledef{C-var}]{
        x : \qtyp{Q}{S} \in \Gamma
      }{
        \Gamma \vdash x : \qtyp{x}{S}
      }

  \vspace{1.1em}
      \infrule[\ruledef{C-app}]{
        \Gamma \vdash s : \cfntype{x}{\qtyp{Q}{S}}{T} \gap \Gamma \vdash Q' \sub Q  \\ \Gamma \vdash t : \qtyp{Q'}{S}
      }{
        \Gamma \vdash \capp{s}{Q'}{t} : T[x \mapstoty Q']
      }
  \end{minipage}\hfill
  \begin{minipage}{0.5\textwidth}
      \vspace{1em}
      \infrule[\ruledef{C-abs}]{
        \Gamma, x : T_1 \vdash t : T_2 \gap \highlight{\Gamma \vdash \vee_{y \in \fv(t) - x}~y \sub P}
      }{
        \Gamma \vdash \fna{P}{x}{t} : \qtyp{P}{\cfntype{x}{T_1}{T_2}}
      }
      
      \vspace{1em}
      \infrule[\ruledef{C-t-abs}]{
        \Gamma, X \sub S \vdash t : T \gap \highlight{\Gamma \vdash \vee_{y \in \fv(t)}~y \sub P}
      }{
        \Gamma \vdash \Sfna{P}{X}{S}{t} : \qtyp{P}{\Sfntype{X}{S}{T}}
      }
    
      \vspace{1em}
      \infrule[\ruledef{C-q-abs}]{
        \Gamma, X \sub S \vdash t : T \gap \highlight{\Gamma \vdash \vee_{y \in \fv(t)}~y \sub P}
      }{
        \Gamma \vdash \Qfna{P}{Y}{Q}{t} : \qtyp{P}{\Sfntype{Y}{Q}{T}}
      }
    
  \end{minipage}
  
    \caption{Typing rules for \Fc}
    \label{fig:capt:typing}

\end{figure}

\paragraph*{Subqualification} One essential change is that we need to adjust subqualification to account for qualifier variables bound by term binders in addition to qualifier variables bound by
qualifier binders.  These changes are the addition of two new rules, \ruleref{sq-refl-tvar} and \ruleref{sq-tvar}.  Rule \ruleref{sq-refl-tvar} accounts for reflexivity in \Fc's
adjusted subqualification judgment.
\ruleref{sq-tvar} accounts for subqualification for qualifier variables bound by term binders, and formalizes this notion of {\it less dangerous} we discussed earlier---that {\tt fifty\_fifty} can be used in a context that allows the use of {\tt one\_ring}, and that {\tt hello\_world} can be used in a context that allows the use of {\tt print}.  Interestingly, it is just a close duplicate of the existing subqualification rule for qualifier variables, \ruleref{sq-var}!

    {\small
    \infrule[]{
        \qual{fifty\_fifty}: \qual{one\_ring}~{\color{lightgray}\tt Unit => Unit}\in \Gamma  \gap  \Gamma \vdash \qual{one\_ring} \sub  \qual{one\_ring}
    }{
        \Gamma \vdash \qual{fifty\_fifty} \sub \qual{one\_ring}
    }}

\paragraph*{Subtyping}
    As function binders introduce a qualifier
    variable, so do function types as well; for example, $x$ in $\cfntype{x}{\qtyp{Q}{S}}{\qtyp{x}{S}}$.  Subtyping needs to account for this bound qualifier variable; see \ruleref{C-sub-arrow}.
    
\paragraph*{Typing}
Values are now qualified with the free variables that they close over (i.e., that they capture).  
To ensure this is faithfully reflected in the value itself, we check that the tag on the value super-qualifies the free variables that value captures. 
This is reflected in the modified typing rules for typing abstractions: \ruleref{C-abs}, \ruleref{C-t-abs}, and \ruleref{C-q-abs}.  
The only other apparent changes are in the rules for term application typing and variable typing. 
While those rules look different, they reflect how term abstractions are a combination of qualifier and term abstractions, and in that setting are no different than the standard rules for typing term variables, term application, and qualifier application! 
These changes to the typing rules are reflected in Figure \ref{fig:capt:typing}.

\paragraph*{Soundness}
Again, we can prove the standard soundness theorems for \Fc, using similar techniques as \citet{10.1145/3605156.3606454}.
\begin{theorem}[Preservation for \Fc]
Suppose $\Gamma \vdash s : T$, and $s \reduces t$.  Then $\Gamma \vdash t : T$ as well.
\end{theorem}
\begin{theorem}[Progress for \Fc]
Suppose $\varnothing \vdash s : T$.  Either $s$ is a value, or $s \reduces t$ for some term $t$.
\end{theorem}

\noindent 
In addition, we recover a {\it prediction lemma} \cite{10.1145/3486610.3486893,odersky2022scoped,boruchgruszecki2021tracking} relating how the free variables of values
relate to the qualifier annotated on their types; in essence, that the qualifier
given on the type contains the free variables present in the value {\tt v}.

\begin{lemma}[Capture Prediction for \Fc]
    Let $\Gamma$ be an environment and $v$ be a value such that $\Gamma \vdash v : \qtyp{Q}{S}$.  Then  $\Gamma \vdash \left\{ \bigvee_{y \in \fv(v)} y\right\} \sub Q$.
\end{lemma}

\section{Mechanization}
\label{section:mechanization}
The mechanization of \Fq (from Section \ref{section:fq}), its derived calculi, \Fm, \Fa, and \Fc, (from Section \ref{section:applications}), and extended \Fq (from Section \ref{section:multiple}), is derived from the mechanization of \Fsub by \citet{10.1145/1328438.1328443}, with some inspiration taken from the mechanization of \citet{10.1145/3605156.3606454} and \citet{Lee2023ToAppear}.
All lemmas and theorems stated in this paper regarding these calculi have been formally mechanized, though our proofs relating
the subqualification structure to free lattices are only proven in text, as we have found Coq's tooling for universal algebra lacking.

\section{Type polymorphism and Qualifier polymorphism}
\label{section:poly}
We chose to model polymorphism separately for qualifiers and simple types.
We introduced a third binder, qualifier abstraction, for enabling polymorphism
over type qualifiers, orthogonal to simple type polymorphism. 
An alternate approach one could take to design a language which needs to model polymorphism over type qualifiers is to have type variables range over {\it qualified types}, that is, types like {\tt \qual{mutable} Ref[Int]} as well as {\tt \qual{const} Ref[Int]}.  This approach can been seen in systems like \citet{10.1145/1103845.1094828, ZibinPLAE2010,Lee2023ToAppear}. However, it also comes with its difficulties: how do we formally interpret repeated applications of type qualifiers?  For example, with a generic {\tt inplace\_map} which maps a function
over a reference cell?

\begin{lstlisting}[language=Scala]
    case class Ref[X](var elem: X)
    // Is this well formed?
    def inplace_map[X](r: mutable Ref[X], f: const X => X): Unit = {
        r.elem = f(r.elem);
    }
\end{lstlisting}

\noindent

For example, what if {\tt inplace\_map} is applied on a {\tt Ref[\qual{const} Ref[Int]]}?  Then {\tt inplace\_map} would expect a function {\tt f} with type {\tt (\qual{const} (\qual{const} Ref[Int])) => \qual{const} Ref[Int]}.  While our intuition would tell us that {\tt \qual{const} (\qual{const}  Ref[Int])} is really just a {\tt \qual{const}  Ref[Int]}, discharging this equivalence in a proof is not so easy.  
Many systems, like \citet{10.1145/1287624.1287637}'s and \citet{10.1145/1103845.1094828}'s sidestep this issue by explicitly preventing type variables from being further qualified, but this approach prevents functions like {\tt inplace\_map} from being expressed at all.
Another approach, taken by \citet{Lee2023ToAppear}, is to show that these equivalences can be discharged through subtyping rules which {\it normalize} equivalent types. 
However, their approach led to complexities in their proof of soundness and it is unclear if their system admits algorithmic subtyping rules. 

Our proposed approach, while verbose, avoids all these complexities by explicitly keeping
simple type polymorphism separate from type qualifier polymorphism.  We would write
{\tt inplace\_map} as:
\begin{lstlisting}[language=Scala]
    case class Ref[Q, X](var elem: Q X)
    def inplace_map[Q, X](r: mutable Ref[{Q} X], f: const X => Q X): Unit = {
        r.elem = f(r.elem);
    }
\end{lstlisting}
Moreover, we can desugar {\it qualified type polymorphism} into a combination of
simple type polymorphism and type qualifier polymorphism.  We can treat a
qualified type binder in surface syntax as a {\it pair of simple type and type qualifier binders}, and have qualified type variables play double duty as simple type variables
and type qualifier variables, as seen in qualifier systems like \citet{wei2023polymorphic}'s.  So our original version
of {\tt inplace\_map} could desugar as follows:
\begin{lstlisting}[language=Scala]
    def inplace_map[X](r: mutable Ref[X], f: const X => X): Unit = {
        r.elem = f(r.elem);
    } // original
    def inplace_map[Xq, Xs](r: mutable[{Xq} Xs], f: const Xs => Xs): Unit = {
        r.elem = f(r.elem);
    } // desugared ==> X splits into Xq and Xs
\end{lstlisting}
One problem remains for the language designer however: how do type qualifiers interact with qualified type variables? 
In our above example we chose to have the new qualifier annotation {\tt \qual{const} X} strip away any existing type qualifier on {\tt X}; this is the approach that \citet{10.1145/1390630.1390656}'s Checker Framework take.  
Alternatively, we could instead merge the qualifiers together:
\begin{lstlisting}[language=Scala]
    def inplace_map[Xq, Xs](r: mutable[{Xq} Xs], f: {const | Xq} Xs => Xs): Unit = 
    {
        r.elem = f(r.elem);
    } // desugared ==> X splits into Xq and Xs
\end{lstlisting}

\section{Revisiting Qualifier Systems}
\label{section:revisit}
Free lattices have been known by mathematicians since \citet{43df3167-5a81-387d-88d7-2d29cdf1c881}'s time as the proper algebraic structure for modelling lattice inequalities involving formulas with variables---word problems---over arbitrary lattices.  In this light it is somewhat surprising that existing qualifier systems have not taken advantage of that structure explicitly, especially so given that is folklore knowledge in the literature that intersection and union types make the {\it subtyping} lattice a free lattice as well as \citet{Dolan} observed.  Here, we revisit some existing qualifier systems to examine how their qualifier structure compares to the structure we present with the free lattice of qualifiers.

\paragraph{A Theory of Type Qualifiers}
\citet{10.1145/301618.301665}'s original work introduced the notion of type
qualifiers, and gave a system for ML-style let polymorphism using a variant of \citet{DBLP:journals/tapos/OderskySW99}'s HM(X) constraint-based type inference.
Qualifier-polymorphic types in Foster's polymorphic qualifier system are a {\it type
scheme} $\forall \overline{Y}/C.T$ for some vector of qualifier variables $\overline{Y}$
used in qualified type $T$ modulo qualifier ordering constraints in $C$, such as
$Y_1 \sub Y_2$.  However, in their system, constraints cannot involve formulas with qualifier variables $X \sub Y_1 \wedge Y_2$ is an invalid constraint, nor are constraints
expressible in their source syntax for qualifier-polymorphic function terms.

\paragraph{Qualifiers for Tracking Capture and Reachability}
Our subqualification system was inspired by the
subcapturing system pioneered by \citet{10.1145/3618003} for use in their capability tracking system for Scala.  They model sets of free variables coupled with operations for merging sets together.  Sets of variables are exactly joins of variables -- the set $\{a, b, c\}$ can be viewed as the lattice formula $a \vee b \vee c$,
and their set-merge substitution operator $\{a, b, c\}[a \mapsto \{d, e\}] = \{d, e, b, c\}$, is just
substitution for free lattice formulas -- $(a \vee b \vee c)[a \mapsto (d \vee e)] = (d \vee e) \vee b \vee c$.
With this translation in mind we can see that they model a free (join)-semilattice, and that their subcapturing
rules involving variables in sets are just translating what the lattice join would be into a set framework.

Independently, \citet{wei2023polymorphic} recently developed a qualifier system for tracking reachability using variable sets as well.  Like \citet{10.1145/3618003}, their subqualification system models a free join-semilattice, with one additional wrinkle.  They model a notion of {\it set overlap}
respecting their subcapturing system as well as a notion of {\it freshness} in their framework to ensure
that the set of values reachable from a function are disjoint, or fresh, from the set of values
reachable from that function's argument.  While {\it overlap} exists only at the metatheoretic level and does
not exist in the qualifier annotations it can be seen that their notion of overlap is exactly the what the lattice {\it meet} of their set-qualifiers would be when interpreted as lattice terms.  Additionally, while freshness unfortunately does not fit in the framework of a free lattice, we conjecture that freshness can be modelled in a setting where lattices are extended with {\it complementation} as well, such as in free complemented
distributive lattices.

\paragraph{Boolean Formulas as Qualifiers}
\citet{10.1145/3485487} recently investigated modelling {\it nullability} as a type qualifier.  
Types in their system comprise a scheme of type variables $\overline{\alpha}$ and Boolean variables $\overline{\beta}$ over a pair of simple type $S$ and Boolean formula $(S, \phi)$,  where values of a qualified type $(S, \phi)$ are nullable if and only if $\phi$ evaluates to {\tt true}.\footnote{Technically they model a triple $(S, \phi, \gamma)$ where $\gamma$ is another Boolean formula which evaluates to {\tt true} if values of type $(S, \phi, \gamma)$ are non-nullable.}  
Boolean formulas form a Boolean algebra, and Boolean algebras are just complemented distributive lattices, so Boolean formulas over a set of variables $\overline{\beta}$ are just free complemented distributive lattices generated over variables in $\overline{\beta}$.  In this sense, we can view \citet{10.1145/3485487}
as a ML-polymorphism style extension of \citet{10.1145/301618.301665}'s
original work which solves Foster's original problem of encoding qualifier constraints:
one can just encode them using Boolean formulas in \citet{10.1145/3485487}'s system.

Unfortunately they do not model subtyping over their qualified types $(S, \phi)$; it would be sensible to say $(S, \phi) \sub (S, \phi')$ if $\phi \implies \phi'$.
They conjecture that such a subtyping system would be sound however.  
While we cannot answer this conjecture definitively, as we only model free lattices, not free complemented distributive lattice systems, it would be interesting future work to extend our framework and theirs
to see if a system modelling free complemented distributive lattice systems with subqualification is sound.

\paragraph{Reference Immutability for C\# \cite{10.1145/2384616.2384619}}
Of existing qualifier systems, the the polymorphism structure of \citet{10.1145/2384616.2384619} is
closest to \Fq. Polymorphism is possible over {\it both} mutability qualifiers and simple types in Gordon's system, but must be done separately, as in \Fq.  The {\tt inplace\_map} function that we discussed earlier would be expressed with both a simple type variable as well as with a qualifier variable:
\begin{lstlisting}[language=Scala]
    def inplace_map[Q, X](r: mutable Ref[{Q} X], f: readonly X => {Q} X): Unit
\end{lstlisting}
Gordon's system also allows for mutability qualifiers to be merged using an operator {\tt \textasciitilde>}.
For example, a polymorphic read function {\tt read} could be written as the following in Gordon's system:
\begin{lstlisting}[language=Scala]
    def read[QR, QX, X](r: {QR} Ref[{QX} X]): {QR ~> QX} X = r.f
\end{lstlisting}

Now, {\tt \textasciitilde>} acts as a restricted lattice join. Given two concrete mutability qualifiers C and D, {C \tt \textasciitilde> D} will reduce to the lattice join of $C$ and $D$.  However, the only allowable judgment in Gordon's system for {\tt \textasciitilde>} when qualifier variables are present, say {\tt C \textasciitilde> Y}, is that it can be widened to \qual{readonly}.

\paragraph{Reference Immutability for DOT \cite{dort_et_al:LIPIcs:2020:13175}}
 roDOT extends the calculus of Dependent Object Types  \cite{amin2016essence} with support for reference immutability.  In their system, immutability
constraints are expressed through a type member field $x.M$ of each object, where $x$ is mutable if and only if $M \leq \bot$, and $x$ is read-only if and only if $M \geq \top$. 
As $M$ is just a Scala type member, $M$ can consist of anything a Scala type could consist of, but typically it consists of type meets and type joins of $\top$, $\bot$, type variables $Y$, and the mutability members $y.M$ of other Scala objects $y$.

While this may seem odd, we can view $M$ as a type qualifier member field of its containing object $x$; the meets and joins in roDOT's $M$'s subtyping lattice correspond to meets and joins in \Fq's subqualification lattice. 
In this sense we can view type polymorphism in roDOT as a combination of polymorphism over simple types and type qualifiers in \Fq. A type $T$ in roDOT breaks down into a pair of a simple type $T \setminus M$ -- $T$ without its mutability member $M$ and $M$ itself.  In this sense \citet{dort_et_al:LIPIcs:2020:13175} provide a different method to encode subqualification; they encode it in type members $M$ and reuse the subtyping lattice
to encode the free lattice structure needed to deal with qualifier polymorphism and 
qualifier variables.

\section{Related Work}
\label{section:related}
\label{sec:related}
\subsection{Languages with Type Qualifier Systems}
\paragraph*{Rust}
\label{subsection:rust} The Rust community is currently investigating approaches \cite{Wuyts_Scherer_Matsakis} for adding qualifiers to Rust.  Their current proposal is to generalize the notion of qualified
types from being a pair of one qualifier and base type to be a tuple of qualifiers coupled to a base type.
Qualifier abstractions are keyed with the kind of qualifier (\qual{const}, \qual{async}, etc, ...) they abstract
over.  

This is easy to see sound using similar ideas to our proof of simplified \Fq, and avoids the complications around subqualification that free lattices over arbitrary lattices pose.  However
this proposal has proven controversial in the Rust community due the additional syntactic complexity it imposes.  

\paragraph*{OCaml}
The OCaml community \cite{Slater_2023_June,Slater_2023} is investigating adding {\it modes} to types for tracking,
in addition to value shapes, properties like {\it uniqueness}, {\it locality}, and {\it ownership},
amongst others; these modes are essentially type qualifiers.  However, modal polymorphism still remains
an open problem in OCaml.

\paragraph*{Pony} Pony's {\it reference capabilities} \cite{10.1145/2824815.2824816} are essentially type qualifiers on base types
that qualify how values may be shared or used.  Pony has qualifiers for various forms
of uniquness, linearity, and ownership properties.  While Pony has bounded polymorphism over {\it qualified types},
Pony does not allow type variables to be requalified, nor does it have polymorphism over qualifiers.

\subsection{Implementing Type Qualifiers}
The Checker Framework by \citet{10.1145/1390630.1390656} is an extensible framework for adding
user-defined type qualifiers to Java's type system.  
The Checker Framework in general allows for qualifying type variables with types, but in their system there is no relationship between a type variable {\tt X} and a qualified type variable {\tt \qual{Q} X}.
Re-qualifying a type variable strips any existing conflicting qualifier from that type variable and what it is instantiated with.

\subsection{Effect Systems}
\label{section:related:effects}
Effect systems are closely related to type qualifiers.  Traditionally, effect annotations
are used to describe properties of {\it computation}, whereas type qualifiers are used to describe
properties of {\it data}.  In the presence of first-class functions, this distinction is often blurred;
for example, modern C++ refers to \qual{noexcept} as a type qualifier on function types \cite{maurer-2015}, whereas traditionally it would be viewed as an effect annotation.  In contrast to type qualifiers, both {\it effect polymorphism} \cite{10.1145/73560.73564} and the {\it lattice structure of effects} \cite{10.1007/978-3-642-31057-7_13} are
well-studied.  However, the interaction of effect polymorphism with {\it subtyping} and {\it sub-effecting} remains understudied.

Many effect systems use {\it row polymorphism} to handle polymorphic effect variables with a restricted form of sub-effecting by subsets \cite{Leijen_2014}.  As for \citet{10.1007/978-3-642-31057-7_13}, they present a lightweight framework with no {\it effect variables}.  Formal systems studying sub-effecting respecting {\it effect bounds} on {\it effect variables} remain rare, despite Java's exception system being just that \cite[Section 8.4.8.3]{10.5555/2636997}.  Curiously, the two extant formal effect systems with these features share much in common with well-known qualifier systems. For example, \citet{leijen2010convenient}'s sub-effecting system can be viewed as a variant of \citet{10.1145/301618.301665}'s lattice-based subqualification system with HM(X)-style polymorphism.
More interestingly, \citet{DBLP:conf/oopsla/Gariano0S19}'s novel Indirect-Call$\varepsilon$ rule, \citet{wei2023polymorphic}'s reachability rule, and \citet{10.1145/3618003}'s subcapturing rule all model a free join-semilattice (of effects).  
In light of all these similarities, and of \citet{10.1145/3607846}'s recent work modelling effect systems with Boolean formulas, we conjecture
that a system modelling free distributive complemented lattices could be used to present an unifying treatment of both effects and qualifiers in the presence of subtyping, subeffecting, and subqualification.

\section{Conclusion}
\label{section:conclusion}
\label{sec:conclusion}
In this paper, we presented a recipe for modelling higher-rank polymorphism, subtyping, and subqualification in systems with type qualifiers by using the {\it free lattice} generated from an underlying qualifier lattice.
We show how a base calculus like \Fsub can be extended using this structure by constructing such an extension \Fq, and we show how the recipe can be applied to model three problems where type qualifiers are naturally suited---reference immutability, function colouring, and capture tracking.
We then re-examine existing qualifier systems to look at how {\it free lattices} of qualifiers show up, even indirectly or in restricted form.  
We hope that this work advances our understanding of the structure of polymorphism over type qualifiers.

\begin{acks}
We thank Brad Lushman, John Boyland, and Guannan Wei for their useful feedback in reading over early drafts of this work. 
We also thank Ross Willard for his useful insights into free lattices.
This work was partially supported by the Natural Sciences and Engineering Research Council of Canada and by an Ontario Graduate Scholarship.
\end{acks}

\bibliography{main}

\end{document}


\section{Supplementary Material}
\begin{lemma}\label{lemma:free:base}
    Let $\Gamma$ be an well-formed environment mapping qualifier variables to their upper bounds, and 
    let $\mathcal{Q}$ be the set of all qualifier terms as defined in \Fq that are well-formed in $\Gamma$;
    that is, bounded lattice formulas.
    Then the subqualification rules of \Fq give rise to the free lattice modulo the bounds in $\Gamma$.
    Namely, that:
    \begin{enumerate}
        \item $\sub$ defines a lattice over $\mathcal{Q}/\sub$, namely, the set of terms of $\mathcal{Q}$ that are not
            equivalent under the order $\sub$.
        \item If $Q[\overline{X}] \in \mathcal{Q}$ and $R[\overline{X}] \in \mathcal{Q}$ are well-formed 
            lattice equations over variables $\overline{X} \in \Gamma$
            such that $\Gamma \vdash Q[\overline{X}] \sub R[\overline{X}]$, then for any bounded lattice $L$ and any instantiation
            of the variables $\overline{X}$ to $\overline{L_X}$ respecting the bounds 
            in $\Gamma$ we have that $Q[\overline{X} \to \overline{L_X}] \leq R[\overline{X} \to \overline{L_X}]$.
    \end{enumerate}
\end{lemma}
\begin{proof}
    To see that $\sub$ defines a lattice over $\mathcal{Q}/\sub$ we need to show that $\sub$ defines a partial order
    and that $\wedge$ and $\vee$ define the meet and join respectively of any two terms in $\mathcal{Q}$ with respect
    to $\sub$.

    By definition $\sub$ is anti-symmetric over $\mathcal{Q}/\sub$ as we are quotienting out already by terms equivalent
    under $\sub$.  Moreover, one can show by induction (as we show in our Coq formalization) that $\sub$ is reflexive and
    transitive as well.  Hence $\sub$ defines a partial order over $\mathcal{Q}/\sub$.  Now, to see that $\sub$ defines
    a lattice, we need to show that $\wedge$ and $\vee$ define the meet and join respectively of any two terms in $\mathcal{Q}$ with respect to $\sub$.  Now, to see that for two lattice terms $Q$ and $R$ that $Q \vee R$ and $Q \wedge R$ are
    lower and upper bounds for both $Q$ and $R$ one can apply the join introduction and the meet elimination rules.  
    It remains to show that they are the greatest lower bounds and the least upper bounds -- but that follows directly from an application of \ruleref{sq-join-elim} and \ruleref{sq-meet-intro}, as desired.

    Finally, we need to show that for two well-formed lattice equations $Q[\overline{X}] \in \mathcal{Q}$ and $R[\overline{X}] \in \mathcal{Q}$ over variables $\overline{X} \in \Gamma$ and $\overline{X} \in \Gamma$
    such that $\Gamma \vdash Q[\overline{X}] \sub R[\overline{X}]$, then for any bounded lattice $L$ and any instantiation
    of the variables $\overline{X}$ to $\overline{L_X}$ and $\overline{X}$ to $\overline{L_X}$ respecting the bounds 
    in $\Gamma$ we have that $Q[\overline{X} \to \overline{L_X}] \leq R[\overline{X} \to \overline{L_X}]$.  This we will
    do by structural induction over the subqualification derivation $\Gamma \vdash  Q[\overline{X}] \sub R[\overline{X}]$.
    Most of the rules hold directly, as they do not involve qualifier variables.  The only two interesting cases are
    the rules \ruleref{sq-var} and \ruleref{sq-refl-var}.  Now, for the \ruleref{sq-refl-var} case we have that for any lattice
    term $l$ that $l \leq l$ so $X[X \to l] \leq X[X \to l]$, as desired.  Finally, for the \ruleref{sq-var} case we have that
    $Y \sub Q \in \Gamma$ and $\Gamma \vdash Q \sub R$.  By induction and assumptions we have that $X$ has been instantiated
    to a lattice term $l \in L$ such that $l \leq Q[\overline{X} \to \overline{L_x}]$, and that $Q[\overline{X} \to \overline{L_X}] \leq R[\overline{X} \to \overline{L_X}]$.  So by transitivity we have that $l \leq R[\overline{X} \to \overline{L_X}]$, hence $X[\overline{X} \to \overline{L_X}] \leq  R[\overline{X} \to \overline{L_X}]$, as desired.
\end{proof}
\begin{lemma}\label{lemma:free:extended}
    Let $M$ be an arbitrary base bounded lattice, let 
    $\Gamma$ be an well-formed environment mapping qualifier variables to their upper bounds, and 
    let $\mathcal{Q}$ be the set of all qualifier terms as defined in extended \Fq that are well-formed in $\Gamma$;
    that is, lattice formulas involving lattice values in $M$.
    
    Then the subqualification rules of extended \Fq  give rise to the free lattice of extensions of $M$ modulo the bounds in $\Gamma$.
    Namely, that:
    \begin{enumerate}
        \item $\sub$ defines a lattice over $\mathcal{Q}/\sub$, namely, the set of terms of $\mathcal{Q}$ that are not
            equivalent under the order $\sub$.
        \item If $Q[\overline{X}] \in \mathcal{Q}$ and $R[\overline{X}] \in \mathcal{Q}$ are well-formed 
            lattice equations over variables $\overline{X} \in \Gamma$
            such that $\Gamma \vdash Q[\overline{X}] \sub R[\overline{X}]$, then for any bounded lattice $L$ extending $M$ and any instantiation
            of the variables $\overline{X}$ to $\overline{L_X}$ respecting the bounds 
            in $\Gamma$ we have that $Q[\overline{X} \to \overline{L_X}] \leq R[\overline{X} \to \overline{L_X}]$.
    \end{enumerate}
\end{lemma}
\begin{proof}
    To see that $\sub$ defines a lattice over $\mathcal{Q}/\sub$ we need to show that $\sub$ defines a partial order
    and that $\wedge$ and $\vee$ define the meet and join respectively of any two terms in $\mathcal{Q}$ with respect
    to $\sub$.

    By definition $\sub$ is anti-symmetric over $\mathcal{Q}/\sub$ as we are quotienting out already by terms equivalent
    under $\sub$.  Moreover, one can show by induction (as we show in our Coq formalization) that $\sub$ is reflexive and
    transitive as well.  Hence $\sub$ defines a partial order over $\mathcal{Q}/\sub$.  Now, to see that $\sub$ defines
    a lattice, we need to show that $\wedge$ and $\vee$ define the meet and join respectively of any two terms in $\mathcal{Q}$ with respect to $\sub$.  Now, to see that for two lattice terms $Q$ and $R$ that $Q \vee R$ and $Q \wedge R$ are
    lower and upper bounds for both $Q$ and $R$ one can apply the join introduction and the meet elimination rules.  
    It remains to show that they are the greatest lower bounds and the least upper bounds -- but that follows directly from an application of \ruleref{sq-join-elim} and \ruleref{sq-meet-intro}, as desired.

    Finally, we need to show that for two well-formed lattice equations $Q[\overline{X}] \in \mathcal{Q}$ and $R[\overline{X}] \in \mathcal{Q}$ over variables $\overline{X} \in \Gamma$ and $\overline{X} \in \Gamma$
    such that $\Gamma \vdash Q[\overline{X}] \sub R[\overline{Y}]$, then for any bounded lattice $L$ and any instantiation
    of the variables $\overline{X}$ to $\overline{L_X}$ and $\overline{X}$ to $\overline{L_X}$ respecting the bounds 
    in $\Gamma$ we have that $Q[\overline{X} \to \overline{L_X}] \leq R[\overline{X} \to \overline{L_X}]$.  This we will
    do by structural induction over the subqualification derivation $\Gamma \vdash  Q[\overline{X}] \sub R[\overline{X}]$.
    Most of the rules hold directly, as they do not involve qualifier variables nor partial evaluation; 
    this includes the case for the rule \ruleref{sq-lift}.  The only four interesting cases are
    the rules \ruleref{sq-var}, \ruleref{sq-refl-var}, \ruleref{sq-lift}, and the two evaluation rules \ruleref{sq-eval-elim}
    and \ruleref{sq-eval-intro}.  The first two rules can be shown to hold via the same proof as in Lemma \ref{lemma:free:base}.
    So it remains to deal with the two cases for \ruleref{sq-eval-elim} and \ruleref{sq-eval-intro}.  As these cases
    are symmetric we may assume without loss of generality that we are working with the rule \ruleref{sq-eval-elim}.
    Hence we have $Q$, $Q'$, and $R$, lattice formulas over lattice values in $M$ such that $\Gamma \vdash Q \sub Q'$,
    $\eval Q' = l$ for some lattice value $l \in M$, and $\Gamma \vdash l \sub R$.  By structural induction
    we have that $Q[\overline{X} \to \overline{L_X}] \leq Q'[\overline{X} \to \overline{L_X}]$ and that
    $l \leq R[\overline{X} \to \overline{L_X}]$.  However, as $\eval Q' = l$, we have that $Q'$ is simply the lattice
    meet and join of lattice values in $L$ that evaluate to $l$; hence $Q' = l$ in $M$.  Hence we have that
    $Q[\overline{X} \to \overline{L_X}] \leq R[\overline{X} \to \overline{L_X}]$, as desired.
\end{proof}